\documentclass[12pt]{article}
\usepackage{authblk}
\usepackage{geometry}
\usepackage{amsthm}
\usepackage{mathrsfs}
\usepackage{subfigure}
\usepackage{customcommands}
\usepackage{bm}
\usepackage{Nettastyle}

\usepackage{amsmath}
\usepackage{amssymb}
\usepackage{amsfonts}
\usepackage{braket}
\usepackage[normalem]{ulem}
\usepackage{color}
\usepackage{bbold}
\usepackage{ulem}
\usepackage{mathtools}
\usepackage{tensor} 

\DeclareMathOperator{\tr}{tr}

\title{Quantum Maximin Surfaces}

\author[1]{Chris Akers,}
\author[1]{Netta Engelhardt,}
\author[2]{Geoff Penington,}
\author[3]{and Mykhaylo Usatyuk}

\affiliation[1]{Center for Theoretical Physics,\\
Massachusetts Institute of Technology, Cambridge, MA 02139, USA}
\affiliation[2]{Stanford Institute for Theoretical Physics,\\ Stanford University, Stanford, CA 94305 USA}
\affiliation[3]{Center for Theoretical Physics and Department of Physics,\\
University of California, Berkeley, CA 94720, U.S.A. and}
\emailAdd{cakers@mit.edu}
\emailAdd{engeln@mit.edu}
\emailAdd{geoffp@stanford.edu}
\emailAdd{musatyuk@berkeley.edu}

\abstract{We formulate a quantum generalization of maximin surfaces and show that a quantum maximin surface is identical to the minimal quantum extremal surface, introduced in the EW prescription. We discuss various subtleties and complications associated to a maximinimization of the bulk von Neumann entropy due to corners and unboundedness and present arguments that nonetheless a maximinimization of the UV-finite generalized entropy should be well-defined. We give the first general proof that the EW prescription satisfies entanglement wedge nesting and the strong subadditivity inequality. In addition, we apply the quantum maximin technology to prove that recently proposed generalizations of the EW prescription to nonholographic subsystems (including the so-called ``quantum extremal islands'') also satisfy entanglement wedge nesting and strong subadditivity. Our results hold in the regime where backreaction of bulk quantum fields can be treated perturbatively in $G_{N}\hbar$, but we emphasize that they are valid even when \textit{gradients} of the bulk entropy are of the same order as variations in the area, a regime recently investigated in new models of black hole evaporation in AdS/CFT.
}

\begin{document}
\maketitle
\section{Introduction}\label{sec:Intro}
Holographic entanglement entropy has served as a valuable tool for decoding numerous facets of the AdS/CFT correspondence, from subregion/subregion duality~\cite{CzeKar12, Wal12, DonHar16} to holographic quantum error correction~\cite{AlmDon14, Har16, HaPPY}. Recently, it has become clear that highly nontrivial aspects of the correspondence manifest under inclusion of perturbative quantum backreaction on the bulk geometry. Non-complementary recovery of the bulk~\cite{HayPen18, AkeLei19}, critical aspects of approximate quantum error correction~\cite{HayPen18}, and most recently the computation of a unitary Page curve for evaporating black holes~\cite{Pen19, AEMM} are all examples of phenomena that show up exclusively in the quantum regime. 

The EW prescription for the quantum-corrected holographic entanglement entropy~\cite{EngWal14} calls for an extremization of a ``quantum-corrected area'' functional, the generalized entropy~\cite{Bek72}:
\be
S_{\mathrm{gen}}= \braket{S_{\mathrm{grav}}} +S_{\mathrm{vN}},
\ee
where $S_{\mathrm{gen}}$ is evaluated on some surface $\sigma$; $S_{\mathrm{grav}}$ is the higher-derivative corrected area (the generalization of the Wald entropy~\cite{IyeWal94}), and $S_{\mathrm{vN}}$ is the von Neumann entropy on one side of $\sigma$ (more precise definitions to follow in Section~\ref{sec:SemiclassicalPreliminaries}).
We write $\braket{S_\mathrm{grav}}$ as an expectation value to emphasize that it consists of curvature operators, which must be smeared to be well-defined~\cite{Lei17, AkeCha17}. This will be discussed at length in the body of the paper.
A surface that extremizes this quantity is called a \textit{quantum extremal surface}~\cite{EngWal14}. The EW prescription then states that
\be
-\tr \rho_{R}\ln \rho_{R} = S_{\mathrm{vN}}[\rho_{R}] = S_{\mathrm{gen}}[{\cal X}_{R}],
\ee
where $\rho_{R}$ is the reduced density matrix of a boundary region $R$ and ${\cal X}_{R}$ is the quantum extremal surface (QES) homologous to $R$. If there is more than one QES, we pick the one with minimal generalized entropy. The entanglement wedge of $R$, denoted $W_{E}[R]$, is then the region which is in between ${\cal X}_{R}$ and $R$ and spacelike to ${\cal X}_{R}$. See Figure~\ref{fig:EntanglementWedge}. Subregion/subregion duality for the EW prescription is the hypothesized equivalence between $\rho_R$ and $W_{E}[R]$.\footnote{Here we really mean that $\rho_R$ is dual to the state of all bulk quantum fields in the entanglement wedge (including the metric).} 

\begin{figure}
\centering
\includegraphics{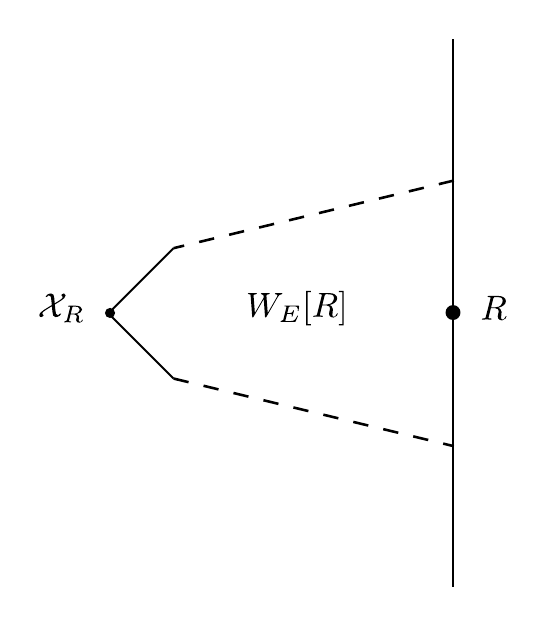}
\caption{The quantum extremal surface ${\cal X}_{R}$ for boundary subregion $R$ with associated entanglement wedge $W_{E}[R]$. The dashed lines are spacelike as a result of caustics, which occur generically. }
\label{fig:EntanglementWedge}
\end{figure}

Consistency requires that the dominant QES satisfy at least three requirements: (1) always lie behind the appropriate causal horizon, (2) boundary CFT entropy of three regions $R_{1},\ R_{2}, \ R_{3}$ computed in this way satisfies the strong subadditivity (SSA) inequality:
\be
S_{\mathrm{vN}}[R_{1}R_{2}]+S_{\mathrm{vN}}[R_{2}R_{3}]\geq S_{\mathrm{vN}}[R_{1}R_{2}R_{3}]+S_{\mathrm{vN}}[R_{2}]~,
\ee
that is, the generalized entropies must satisfy:
\be
S_{\mathrm{gen}}[{\cal X}_{R_{1}R_{2}}]+S_{\mathrm{gen}}[{\cal X}_{R_{2}R_{3}}]\geq S_{\mathrm{gen}}[{\cal X}_{{}R_{1}R_{2}R_{3}}]+S_{\mathrm{gen}}[{\cal X}_{R_{2}}]~,
\ee
and (3) the entanglement wedges of two nested boundary regions $R'\subset R$ also nest: $W_{E}[R']\subset W_{E}[R]$. 

Let us spend a moment explaining the reasoning for (1) and (3), as (2) is self-evident. If (1) were false, then it would be possible to modify the entanglement entropy of $\rho_{R}$ by turning on some local unitary operator, in violation of invariance of $S_{\mathrm{vN}}$ under unitary transformations. (3) is a bit more subtle: since subregion/subregion duality conjectures an isomorphism between bulk and boundary operator algebras, and since $R'\subset R$ implies that the operator algebras are also nested under inclusion ${\cal A}_{R'}\subset {\cal A}_{R}$, it immediately implies that the same must be true of the bulk. But clearly if $W_{E}[R']$ contains points that are not  in $W_{E}[R]$, then there are local operators in the operator algebra of $R'$ that are not in the operator algebra of $R$. 

Point (1) was proven in~\cite{EngWal14} using the Generalized Second Law (GSL), which states that the generalized entropy is nondecreasing towards the future along causal horizons (see~\cite{Wal10Proofs} for a review). Points (2) and (3), however, remain unproven. When it was thought that in the semiclassical regime the dominant QES always lies within a Planck length of the minimal classical extremal surfae (the HRT surface~\cite{RyuTak06, HubRan07}), this seemed less problematic, since any violations would at most be Planck-scale and only appear in non-generic cases where the classical spacetime saturates the conditions. However, it was recently shown~\cite{Pen19, AEMM} that QESs can lie arbitrarily far from classical extremal surfaces. This puts a new urgency on proving (2) and (3) for QESs. 

The reason that the proofs have remained missing for so long  -- the EW prescription is now over five years old -- is that at the classical level the proofs require the vastly different machinery of Wall's \textit{maximin} construction~\cite{Wal12}, an alternative (but equivalent) formulation of the HRT proposal. Maximin surfaces are defined operationally in the following way: given a boundary region $R$, we consider all bulk Cauchy slices that contain $R$.\footnote{In the original maximin proposal~\cite{Wal12} all Cauchy slices that contained $\partial R$ were allowed in the maximinimization. In the restricted maximin proposal~\cite{MarWal19} only Cauchy slices fixed to a chosen boundary slice $C_{\partial}$, with $R \subset C_{\partial}$, were maximinimized over.} On each such Cauchy slice $C$, we find the minimal area surface homologous to $R$, denoted $\min [C, R]$. We then maximize the area of $\min[C, R]$ over the set of all Cauchy slices $C$; the resulting surface is the so-called maximin surface of $R$. Under certain assumptions, Wall then argues that the resulting surface is the HRT surface of $R$. The minimality of the maximin on some Cauchy slice is a powerful tool that, in conjunction with the Null Convergence Condition ($R_{ab}k^{a}k^{b}\geq 0$ for null $k^{a}$), NCC for short, can be used to show that HRT surfaces satisfy SSA and that their entanglement wedges nest~\cite{Wal12, HeaTak, EngHarTA}.\footnote{The argument uses maximin technology and focusing to reduce the proof of the covariant formulation to the proof in the static case~\cite{HeaTak,EngHarTA}.}

Partly to close the nesting and SSA gap in the QES literature, partly to justify a component of the arguments of~\cite{Pen19, BouCha19}, and partly for newfound applications to nonholographic systems coupled to the bulk, we develop the quantum analogue of maximin surfaces -- quantum maximin surfaces. These are operationally defined in much the same way as the original maximin, but with the generalized entropy replacing the area as the quantity to be minimized and subsequently maximized. 

Need a quantum maximin surface always exist? Even at the classical level, this is not known in full rigor and broad generality. At the quantum level, new subtleties appear involving UV divergences and oscillations of surfaces on transplanckian scales. We discuss these various subtleties in Section~\ref{sec:existence} and conclude that nonetheless the overwhelming evidence is in favor of the existence of such surfaces. 

The proofs of nesting and SSA do not generalize quite immediately from the classical maximin case, especially in situations (which we do \textit{not} exclude) in which the bulk is coupled to an external system and allowed to evolve as a mixed state on its own. In adapting the arguments to the quantum regime, we necessarily need to replace the NCC with an analogous statement that guarantees a quantum analogue of gravitational lensing -- the quantum focusing conjecture (QFC)~\cite{BouFis15}. However, a naive application of the latter to situations in which the bulk is allowed to exchange radiation with an external system yields wrong results. Great care must be taken in applying the QFC to such setups. We must also make use, in both proofs, of strong subadditivity of bulk entropy. 

We additionally investigate another inequality on the classical entropy cone: monogamy of mutual information (MMI). This is the inequality:

\begin{equation}
	S_{\mathrm{vN}}[R_1 R_2] + S_{\mathrm{vN}}[R_1 R_3] + S_{\mathrm{vN}}[R_2 R_3] \ge S_{\mathrm{vN}}[R_1] + S_{\mathrm{vN}}[R_2] + S_{\mathrm{vN}}[R_3] + S_{\mathrm{vN}}[R_1 R_2 R_3]~,	
\end{equation}
for boundary subregions $R_i$. This has been established by~\cite{HayHeaMal13} in the classical static case and by the maximin formalism in the covariant formulation. However, it is an inequality which is not in general true of quantum field theories. We therefore should not expect it necessarily to hold for holographic states once we include perturbative quantum corrections. Indeed, it is easy to show that boundary MMI can be violated if the bulk quantum fields themselves violate MMI.
We point out that the converse is less obvious. While proving boundary SSA involves an inequality on three bulk regions, which turns out to be bulk SSA, the inequality for boundary MMI involves {\it seven} bulk regions.
We leave to future work the question of whether the assumption of bulk MMI is sufficient to derive this seven-party inequality (and more generally what assumptions on bulk entropies are sufficient for any holographic inequality satisfied classically to remain valid under bulk quantum corrections).

The quantum maximin formalism turns out to be powerful enough to accommodate a new twist on the EW prescription that includes a nonholographic external system coupled to our holographic system. This idea was originally introduced in \cite{HayPen18}, where it was shown that a rigorous derivation of entanglement wedge reconstruction, using the full machinery of approximate quantum error correction, requires considering quantum extremal surfaces for the \emph{combination} of a boundary region $R$ and an entangled nonholographic auxiliary system $Q$.\footnote{We use notation here consistent with the rest of this paper. Slightly confusingly, in \cite{HayPen18}, the boundary region was denoted by $\bar{A}$, while the auxiliary system was denoted by $R$.}  

Surprisingly, as shown in \cite{Pen19, AEMM}, this new type of quantum extremal surface may be nonempty, even when the boundary region $R$ is empty (i.e. we only look at the entangled nonholographic system $Q$). This idea was developed further in~\cite{AMMZ}, where it was called the ``quantum extremal island conjecture.'' 

Formally, a quantum extremal surface ${\cal X}_{R,Q}$ for the combination of boundary region $R$ and nonholographic system $Q$ is a codimension-two bulk surface that is an extremum of a generalization of generalized entropy, which we shall dub the \textit{hybrid entropy}:
\be
S_{\mathrm{hyb}}[R,Q]\equiv \mathrm{ext}_{{\cal X}_{R,Q}} \left [S_{vN}[H_{R,Q} \cup Q] + \braket{S_{\mathrm{grav}}[{\cal X}_{R,Q}]} \right]~,
\ee
where $\partial H_{R,Q} = {\cal X}_{R,Q} \cup R$ and $S_{vN}[H_{R,Q} \cup Q]$ is the von Neumann entropy of $Q$ together with bulk fields in $H_{R,Q}$.\footnote{This formula generalizes the formulas given in, for example, Eqn. (4.14) of \cite{HayPen18} and Eqn. (15) of ~\cite{AMMZ} to include higher derivative corrections.} If $R$ is empty, $H_{R,Q}$ forms an `island' in spacetime, bounded by ${\cal X}_{R,Q}$.
The proposal of \cite{HayPen18,Pen19,AMMZ} is that the entropy of $R \cup Q$ is given by the hybrid entropy of the minimal-hybrid-entropy quantum extremal surface. For certain classes of matter fields and two bulk spacetime dimensions, a replica argument was recently used to argue for the ``island'' formula ~\cite{PenShe19, AlmHar19}.

For this prescription to have any hope of being successful (especially in more than two spacetime dimensions, where it is conjectured to hold~\cite{AlmMahSan19}), it must be the case the the hybrid entropy satisfies SSA. Because the new proposals also suggest that the combination of $R$ and $Q$ encodes an entanglement wedge, nesting is necessary as well. By implementing a small modification of the definition of a quantum maximin surface, we are able to show that this is indeed the case: extrema of the hybrid entropy do satisfy both SSA and nesting. 

The paper is structured as follows. In Section \ref{sec:2}, we review the necessary semiclassical gravity preliminaries, define notation, and discuss the types of bulk boundary conditions that we will consider. We comment on subtleties involved in evaluating the generalized entropy of surfaces with open and absorbing boundary conditions.

In Section \ref{sec:3}, we start by reviewing the classical maximin construction. We then define and argue for the existence of quantum maximin surfaces and show that they are equivalent to quantum extremal surfaces.

In Section \ref{sec:applications}, we use the quantum maximin construction to prove SSA, nesting, and that the entanglement wedge contains the causal wedge. We also comment on other holographic entropy inequalities, providing a counterexample to MMI at next to leading order.

In Section \ref{sec:nonholographic}, we extend the quantum maximin construction to the case where nonholographic degrees of freedom are entangled with a holographic system. We argue that the modified quantum maximin surface exists, and prove that it is equivalent to the hybrid entropy extremal surface. Using a slightly modified quantum maximin construction we prove SSA and EWN for the hybrid entropy, providing a crucial consistency check on the conjecture. 

In the appendix, we comment on subtleties related to the existence of quantum maximin surfaces in the presence of higher curvature corrections. We focus on situations where subleading terms in the gravitational entropy compete with the area term, and argue that the naive gravitational entropy is not well defined for surfaces with transplanckian fluctuations even when all extrinsic curvatures are small relative to the Planck scale.

\section{Preliminaries}\label{sec:2}
In this section, we describe various conventions, assumptions, and concepts that will be used throughout. Any conventions not explicitly stated are as in~\cite{Wald}.

\subsection{Bulk and Semiclassical Gravity Preliminaries}
\label{sec:SemiclassicalPreliminaries}

We work on a smooth manifold $M$ with a $C^{2}$ background metric $g^{(0)}$. Throughout the rest of this paper, we assume the semiclassical expansion, in which we include the backreaction due to quantum fields propagating on our spacetime; as a consequence the spacetime metric admits a description as a perturbative series in some small parameter around the fixed classical background $g^{(0)}_{ab}$:
\be \label{eq:semiclassical}
g_{ab} = g^{(0)}_{ab} + g^{(1/2)}_{ab}+ g^{(1)}_{ab}+ g^{(3/2)}_{ab}+\cdots
\ee
where $g^{(n)}_{ab}$ is the term in the expansion which is of order $n$ in the small parameter. (The fractional powers are a result of graviton contributions.) We will consider an expansion in both $\alpha'$ and $G_{N}\hbar$, including both higher derivative and quantum corrections, since the latter normally incurs the former. We assume that the spacetime $(M,g)$ is stably causal and connected (but possibly with multiple asymptotic boundaries). We will generally assume that $(M, g^{(0)})$ is asymptotically locally AdS unless otherwise stated.

We will largely work with the conventions of~\cite{EngFis19}, where a surface is defined as an achronal, embedded, codimension-two submanifold $\sigma$.\footnote{That is, a surface is an embedding map $f:U\rightarrow M$ from a manifold $U$ into $M$ such that the image $\sigma\equiv f(U)$ in $M$ is spacelike and achronal.} Since $(M,g)$ is stably causal, we may define a time function on it whose level sets are constant time slices. In an abuse of notation (since we are not assuming global hyperbolicity), we refer to such slices as Cauchy slices. A surface is said to be \textit{Cauchy-splitting} if it divides a Cauchy slice $C$ into two disjoint components $C_{\mathrm{in}}$ and $C_{\mathrm{out}}$, where $C=C_{\mathrm{in}}\cup \sigma\cup C_{\mathrm{out}}$. We will add Cauchy splitting to the definition of a surface, so that all surfaces discussed are Cauchy-splitting. We define a hypersurface (by contrast with a surface) as an embedded codimension-one submanifold. 

We will be interested in the EW prescription, which is the quantum generalization of the HRT proposal for computing the entropy of boundary subregions. To this end, let us first state our conventions about the boundary and then review the HRT and EW proposals in detail.  We will assume that the metric on $\partial M$ is geodesically complete\footnote{In any conformal class it is possible to engineer a pathological conformal factor that makes $\partial M$ geodesically incomplete; we assume that there exists a representative in the conformal class of the boundary metric where this is not true, and we work with that conformal factor.} and smooth. We define a \textit{region} $R$ on the boundary as a finite union of smooth connected codimension-one subsets of $\partial M$ where the entire union is acausal. 

We will say that a surface $\sigma$ is \textit{homologous} to a boundary region $R$ if there exists an acausal hypersurface $H$ such that $\partial H = \sigma\cup R$. We call $H$ the homology hypersurface of $\sigma$ to $R$~\cite{EngHarTA}. If $\sigma$ is the quantum extremal surface picked out by RT/HRT/EW prescriptions~\cite{RyuTak06, HubRan07, EngWal14}, (see below for definitions), we denote $H$ as $H_{R}$, and the domain of dependence $D[H_{R}]$ is the entanglement wedge of $R$, denoted $W_{E}[R]$.

Let us now remind the reader of the HRT proposal. The proposal relates the von Neumann entropy $S_{\mathrm{vN}}=-\tr \rho_{R} \ln \rho_{R}$ of a boundary subregion $R$ to the area of the minimal area \textit{extremal} surface $\sigma_{R}$ homologous to a boundary region $R$. Here \textit{extremal} means that the mean curvature vector $K^{a}$ vanishes. This vector is the trace of the extrinsic curvature, and is defined for an arbitrary surface $\sigma_{R}$ as:
\be \label{eq:classicalcurvature}
K^{a}\equiv K^{a}_{bc}h^{bc} = -h^{bc} h^{d}_{b}h^{e}_{c} \nabla_{d}h_{e}^{a}
\ee
(here $h_{ab}$ is the induced metric on $\sigma_{R}$). Physically speaking, the sign of $K^{a}n_{a}$ computes whether the area of $\sigma_{R}$ increases or decreases as we deform $\sigma_{R}$ along some vector field $n_{a}$. For a null direction $k^{a}$, the contraction $K^{a}k_{a}[\sigma_{R},y]$ is called the null expansion of $\sigma_{R}$ at internal coordinate $y$, and it takes the form
\be 
\theta[\sigma,y]_{k}=K^{a}[\sigma,y]k_{a}\propto \frac{d \mathrm{\delta A}}{d\lambda}
\ee
where $\delta A$ is an infinitesimal area element and $\lambda$ is the affine parameter along $k^{a}$. Thus the sign of $\theta[\sigma,y]_{k}$ encodes whether the area is increasing or decreasing as we evolve from $\sigma$ along $k^{a}$. As explained above, an extremal surface is defined as a surface where $K^{a}=0$ everywhere. To reiterate, the \textit{HRT surface} is the minimal area extremal surface homologous to $R$. 

The quantum generalization of this prescription follows a long tradition of replacing the area in a classical statement by the generalized entropy to obtain a quantum analogue. Let us briefly review the motivation behind this replacement. Standard theorems about areas of surfaces such as the focusing theorem or the Hawking area theorem ordinarily use the NCC, and in particular the null energy condition ($T_{ab} k^{a}k^{b}\geq 0$, which is equivalent to the NCC when the semiclassical Einstein equation holds), which is routinely violated in quantum field theory~\cite{Haw75, EpsGla65, Cas48, Dav76, DavUnr76, Dav77}. The technique for fortifying such statements against quantum correction calls for a replacement of the area by the so-called generalized entropy (originally by~\cite{Bek73, Haw75}, since then this has been a remarkably successful research program; see~\cite{Wal10QST, EngWal14, BouFis15, BouEng15c} for just a few examples).

Let $\sigma$ be a Cauchy-splitting surface as defined above, homologous to $R$ with homology hypersurface $H$. The \textit{generalized entropy} of $\sigma$ is defined:
\be \label{eq:Sgen}
S_{\mathrm{gen}}[\sigma] = \braket{S_{\mathrm{grav}}[\sigma]} + S_{\mathrm{vN}}[\sigma].
\ee
Here $S_{\mathrm{grav}}[\sigma]$ is the higher derivative functional which replaces the area~\cite{Don13,GuoMia14,Cam13} (similar but more general than the Wald entropy~\cite{IyeWal94}) and $S_{\mathrm{vN}}[\sigma]=-\mathrm{tr} \left ( \rho_{H} \ln  \rho_{H}\right)$ is the von Neumann entropy of the reduced density matrix on $H$.

We now define the quantum mean curvature~\cite{EngFis19}, a generalization of $K^{a}$ that includes quantum corrections by replacing area with generalized entropy.~\footnote{Technically,~\cite{EngFis19} only showed that the quantum mean curvature is well-defined as a distributional tensor for quantum corrections, not including higher derivative corrections; here we will include higher derivative corrections. The functional is not known in general closed form, but we provide examples to illustrate the point.}
\be\label{eqn:qmc}
    {\cal K}_a[\sigma,y] = 4 G_{N} \hbar \frac{{\cal D} S_\mathrm{gen}}{{\cal D} X^a[y]}~.
\ee
where $X^a[y]$ are the embedding coordinates of $\sigma$ and ${\cal D}/{\cal D}X^{a}$ is a \textit{functional covariant derivative}; see~\cite{EngFis19} for details. The quantum expansion (originally defined in~\cite{BouFis15}, though see also~\cite{Wal10QST, EngWal14} for earlier work) can then simply be defined as:
\be
\Theta[\sigma,y]_k = {\cal K}_a[\sigma,y] k^a~.
\ee
This encodes how the generalized entropy changes with deformations of $\sigma$ in the $k^a$ direction. We will further assume a quantum null generic condition about the quantum expansion: that it cannot remain zero for a finite amount of affine parameter.

A \textit{quantum extremal surface} (QES) is a surface ${\cal X}_{R}$ with vanishing quantum mean curvature. The EW prescription relates the von Neumann entropy $S_{\mathrm{vN}}[\rho_{R}]$ of a boundary region $R$ to the quantum extremal surface homologous to $R$ with smallest $S_{\mathrm{gen}}$:
\be
S_{\mathrm{vN}}[\rho_{R}] = S_{\mathrm{gen}}[{\cal X}_{R}]~,
\ee
where ${\cal X}_{R}$ is the minimal $S_{\mathrm{gen}}$ quantum extremal surface homologous to $R$.\footnote{For recent work on proving this prescription see~\cite{DonLew17}.} The entanglement wedge in this quantum corrected prescription is $D[H_{R}]$, the domain of dependence of the homology hypersurface of ${\cal X}_{R}$, which in general will not coincide with the entanglement wedge defined by the homology hypersurface of the minimal area \textit{classical} extremal surface homologous to $R$. Indeed, the recent interest in QESs stems from the discovery that they can be displaced from the classical (HRT) surface by a large amount~\cite{Pen19, AEMM}, or exist in situations where the HRT surface is the empty set.  

Is the generalized entropy of an arbitrary surface well-defined and UV-finite? This question cannot be answered completely satisfactorily without a direct formulation of nonperturbative quantum gravity. However, there is mounting evidence that the generalized entropy is well-defined and UV-finite for general surfaces (see~\cite{BouFis15} for a review and references within, as well as~\cite{DonMar19} for a recent discussion). Historically, it was initially noted that, for minimally coupled scalars, the renormalization of $G_{N}$ due to radiative corrections was essentially the counterterm to the von Neumann entropy divergence~\cite{FurSol94}. Since then, various theories from non-minimally coupled scalars to spinors have been studied where divergences of the matter entropy were absorbed in renormalization of $G_{N}$ (though notably it is not yet clear how graviton contributions to $S_{\mathrm{vN}}$ compare with the renormalized $G_{N}$, as gravitons are difficult to define off-shell). The higher derivative corrections are likewise significant, as they contain extrinsic curvature terms which are expected to balance contributions to $S_{\mathrm{vN}}$ which may diverge due to e.g. corner divergences (see~\cite{BouFis15, DonMar19} and references therein).

In this work we will make the assumption that the generalized entropy is UV-finite. The more mathematically-inclined reader may find this well-motivated but admittedly unproven statement objectionable; however, we follow a mainstream approach in the literature in which the mounting wall of evidence in favor of the finiteness of generalized entropy is used as a motivator for simply asserting the desired assumption. Another potential complaint is that surfaces do not have a sharply-defined location once we consider perturbations to the geometry. This Planck-scale fuzziness suggests that we should only consider surfaces as localized objects up to Planckian scales. Under the implementation of a UV cutoff at energy scale no higher than the Planck scale, this however leaves the generalized entropy of a surface well-defined. This will be discussed at greater length in Sec.~\ref{sec:existence} for surfaces which minimize the generalized entropy. 

It is also worth noting that even the classical area can be divergent if the surface $\sigma$ is not compact, for example if  the boundary region $R$ has nonempty boundary $\partial R$. This divergence needs to be regulated by cutting off the bulk spacetime at some finite radius. The higher-curvature corrections and bulk von Neumann entropy $S_{\mathrm{vN}} [\sigma]$ will also have IR divergences that are regulated in the same way.

Returning to the quantum expansion $\Theta[\sigma,y]_k$, let us consider a specific example to gain intuition about the relative contribution of the geometric and quantum terms. 
This quantity was computed explicitly in \cite{AkeCha17} for the special case of four-derivative gravity, 
\be
    I_\mathrm{grav} = \frac{1}{16\pi G_{N}} \int \sqrt{g} \big(R + \ell^2 \lambda_1 R^2 + \ell^2 \lambda_2 R_{ab}R^{ab} + \ell^2 \lambda_{GB} {\cal L}_\mathrm{GB}\big)~,
\ee
where ${\cal L}_\mathrm{GB} = R^2_{abcd}-4R^2_{ab}+R^2$ is the Gauss-Bonnet curvature, $\ell$ is the the cutoff length scale of the effective field theory, and $\lambda_1,\lambda_2,\lambda_\mathrm{GB}$ are assumed to be renormalized. 
The generalized entropy functional for this theory coupled to matter has been computed using replica methods \cite{Don13},
\begin{align}
    S_\mathrm{gen}[\sigma] = \frac{A[\sigma]}{4 G_{N}} + \frac{\ell^2}{4 G_{N}}\int_\sigma \sqrt{h} \big[ 2 \lambda_1 R + \lambda_2 \big(R_{ab} N^{ab} - \frac{1}{2} K_a K^a \big) + 2 \lambda_\mathrm{GB} r \big] + S_\mathrm{vN,ren}~,
\end{align}
where $N^{ab}$ is the projector onto the normal space of $\sigma$, $r$ is the intrinsic Ricci scalar of $\sigma$, and $S_\mathrm{vN,ren}$ is the renormalized von Neumann entropy.
With this, the authors of \cite{AkeCha17} computed 
\begin{equation}
\begin{split}
    \Theta[\sigma]_k =& \theta[\sigma]_k + \ell^2\bigg[ 2\lambda_1 ( \theta[\sigma]_k R + \nabla_k R) + \lambda_2 \big( (D_i - \omega_i)^2 \theta[\sigma]_k + K_{a}K^{aij}K_{ij}^k \\
    &+\theta[\sigma]_k R_{klkl} + \nabla_k R - 2 \nabla_l R_{kk} + \theta[\sigma]_k R_{kl} - \theta[\sigma]_l R_{kk} + 2 K^{kij}R_{ij}\big)\\
    &-4 \lambda_\mathrm{GB}\big( r^{ij}K^k_{ij} - \frac{1}{2} r \theta[\sigma_k]\big) \bigg] + 4 G_{N} \frac{k^a}{\sqrt{h}} \frac{{\cal D} S_\mathrm{vN,ren}}{{\cal D} X^a}~,
\end{split}
\end{equation}
where indices $i,j$ are coordinates intrinsic to the surface $\sigma$, indices $k,l$ denote contraction of ambient indices $a,b$ with the respective orthogonal null direction $k$ or its normal $l$, $D_i$ is the covariant derivative intrinsic to $\sigma$, $\omega_i \equiv l_a D_i k^a$ is the normal connection (also known as the twist potential), and $X^a$ are the embedding coordinates of $\sigma$. 
We have suppressed the dependence of every term on the point $y$. 



The quantum expansion obeys a powerful lemma, first proven in \cite{Wal10QST}, which we will use in the proofs of nesting in Sections \ref{sec:applications} and \ref{sec:nonholographic}. 
\begin{lem}\label{Wallemma} (Wall's Lemma)~\cite{Wal10QST}
Let $\sigma,\sigma'$ be two co-dimension two surfaces that contain the point $p \in \sigma \cap \sigma'$ such that they are also tangent at $p$. By definition, both surfaces split a Cauchy surface in two. We will arbitrarily choose one side of $\sigma$ on some Cauchy slice $C$ to be In$[\sigma]$ and take $\sigma'$ to be a surface that does not intersect Out$[\sigma]$ anywhere. We will make a choice of labels of In and Out for $\sigma'$ which is consistent with the choice for $\sigma$:  $D[In[\sigma']]\subset D[In[\sigma]]$. Let $\ell^{a}$ be a future-directed null vector field on $\sigma$, which, when projected onto a Cauchy slice $C$ containing $\sigma$, point towards In$[\sigma]$. We will also define $\ell^{a}$ on $\sigma'$ analogously (if $\sigma$ is the surface in the EW prescription, this is the future generator of the entanglement wedge of one side of $\sigma$). Let $\theta$ represent the null expansion of $\ell^{a}$ at $p$, see Figure~\ref{fig:WallsLemma}. 
In the classical regime, an older proof of Galloway's~\cite{Gal99} shows that
\begin{align}
    \theta[\sigma] \ge \theta[\sigma']~.
\end{align}
In the semiclassical regime, Wall improves this result (and gives an alternative proof of the classical version) to bound the quantum expansions 
\begin{align}
    \Theta[\sigma] \ge \Theta[\sigma']~.
\end{align} \end{lem}
The proof of this quantum result uses weak monotonicity of relative entropy to handle the von Neumann entropy, and it uses the semiclassical expansion to justify ignoring the higher-curvature corrections \cite{EngWal14}.

In fact, we will also need a spacelike analogue of this lemma. Since this spacelike result has not appeared in the literature thus far, we will give a proof of it below.

\begin{lem} \label{SpacelikeLemma}
Let $\sigma$ and $\sigma'$ be two codimension-two surfaces that lie on a single Cauchy slice $C$ where $\sigma'$ lies entirely on one side of $\sigma$ and is tangent to $\sigma$ at a point $p$ (or multiple points). We assume that $C$ is acausal in a neighborhood of $p$. Let $r^{a}$ be a vector field on $C$ which is normal to $\sigma$ and $\sigma'$ and points towards In$[\sigma]$ on $\sigma$ and towards In$[\sigma']$ on $\sigma'$ (where In and Out are defined as above). Then:
\be
\left. r_{a} {\cal K}^{a}[\sigma] \right |_{p} \geq \left . r_{a}{\cal K}^{a}[\sigma'] \right |_{p}.
\ee
\end{lem}

\begin{proof}
We begin with the classical extrinsic curvature component. Because within $C$, $\sigma$ and $\sigma'$ are both codimension-one, the classical extrinsic curvature $K_{ab}[\sigma]$ and $K_{ab}[\sigma']$ is a rank-two tensor. Recall that, when measured with respect to a normal vector pointing towards the \textit{exterior}, $K_{ab}[\sigma]v^{a}v^{b} |_{p}$ measures how much $\sigma$ curves away from its tangent plane at $p$ with motion along $v^{a}$. Because $\sigma$ lies outside of $\sigma'$, $\sigma'$ must curve away from its tangent plane more than $\sigma$ does. So for any vector field $v^{a}$, this implies that with respect to a normal vector field pointing towards In$[\sigma]$ and In$[\sigma']$:
\be
r_{c}K^{c}_{ab}[\sigma']v^{a}v^{b}|_{p} \leq r_{c}K^{c}_{ab}[\sigma]v^{a}v^{b}|_{p}
\ee
This then immediately implies that
\be
r_{c} K^{c}_{ab}[\sigma']q^{ab}|_{p} \leq r_{c} K^{c}_{ab}[\sigma]q^{ab}|_{p}
\ee
where $q^{ab}$ is the metric on $C$ (in this case induced from the Lorentzian spacetime metric, but for our purposes in this proof that is irrelevant). This immediately shows that the classical component of $r_{a}{\cal K}^{a}$ satisfies the desired inequality.

The proof that the quantum contributions to ${\cal K}^{a}$, $\frac{{\cal D} S_{\mathrm{vN}}}{\mathcal{D} \sigma^{a}}$, follows exactly the same line of reasoning as the null proof in Lemma~\ref{Wallemma} (in fact, it is simpler: the null proof requires mapping back to a Cauchy slice; here we simply begin with all regions already lying on a Cauchy slice).  

Finally, we must deal with the higher curvature corrections. We expect that these are in general -- as was assumed in~\cite{Wal10QST} -- subleading to the mean curvature and the entropy contributions, and thus do not alter the conclusions. 

\end{proof}

\begin{figure}
\centering
\includegraphics{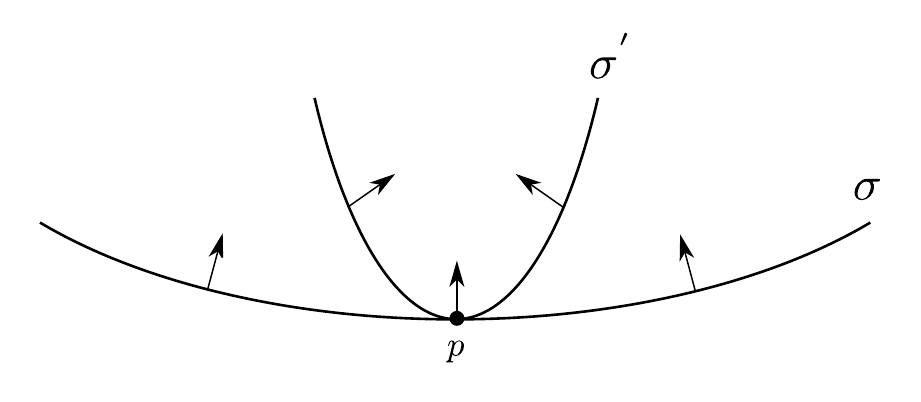}
\caption{The surfaces $\sigma$ and $\sigma'$ are shown tangent at the point $p$. The arrows illustrate the projection of the null orthogonal vectors onto the Cauchy slice. The associated null expansion of $\sigma'$ lowerbounds that of $\sigma$.}
\label{fig:WallsLemma}
\end{figure}

\vspace{3mm}

Finally, let us comment on how various powerful tools in the arsenal of General Relativity and causal structure are generalized to include quantum corrections. One of the most powerful pillars of classical causal structure is the Raychaudhuri equation for null geodesics, which for spacetimes satisfying the NCC implies the focusing theorem:
\be
\frac{d\theta}{d\lambda}\leq - \frac{\theta^{2}}{D-2},
\ee
where  $\lambda$ is an affine parameter along the null geodesics. For $\theta=0$ (e.g. for extremal surfaces), this means that $\theta$ is nonpositive away from the surface along $k^{a}$. This is used in the content of many proofs in gravity and holography in the $G_{N}\hbar\rightarrow 0$ limit: entanglement wedge nesting, strong subadditivity of the HRT proposal, and the fact that entanglement wedges include causal wedges~\cite{Wal12}. To prove analogous statements in the quantum regime, we will need the quantum analogue of the focusing theorem, known as the quantum focusing conjecture (QFC)~\cite{BouFis15}:
\be\label{eqn:qfc}
0 \ge \frac{ {\cal D} \Theta[\sigma,y']_{k}}{ {\cal D}\lambda} = \int_\sigma d^{d-2}y \frac{\delta \Theta[\sigma,y]_k}{\delta X^a(y')}k^a(y') ~.
\ee
Some comments are required. 
First, the functional derivative contains both local terms -- i.e. terms proportional to $\delta(y-y')$ and derivatives of $\delta(y-y')$ -- as well as terms non-local in $y,y'$. 
The gravitational terms only have contributions of the local type (the classical expansion at one point on the surface does not care if you deform the surface somewhere else).
The von Neumann entropy will give terms of both types. The non-local terms are shown to be non-positive from the strong subadditivity of von Neumann entropy \cite{BouFis15}.
The local terms from $S_\mathrm{vN,ren}$ are generally hard to compute, and are related to the quantum null energy condition \cite{BouFis15b,KoeLei15,FuKoe17b,AkeCha17,BalFau17} (see \cite{LeiLev18,BalCha19} for discussion specifically of the local aspects of the quantum null energy condition). Second, it is vital that all geometric quantities be understood as quantum operators, which must be smeared at the effective field theory cutoff scale \cite{FuKoe17,Lei17}.  So really it is the expectation value $\braket{S_\mathrm{grav}}$ that contributes to $\Theta$, and in the inequality \eqref{eqn:qfc} we are implicitly smearing in $y$ all geometric terms over some profile.

Finally, note that when we impose reflecting boundary conditions at infinity, the choice of Cauchy slice is immaterial for evaluating $S_{\mathrm{gen}}$. This is not the case when we impose more general boundary conditions, which we shall now discuss.

\subsection{Boundary Conditions} \label{sec:boundaryconditions}
We close this section with a discussion of boundary conditions. Recent renewed interest in the quantum extremal surface prescription has featured setups with absorbing (or, more generally, coupled) boundary conditions on the asymptotic boundary. With these boundary conditions, it is clear that the Cauchy slice on which the $S_{\mathrm{vN}}$ component of $S_{\mathrm{gen}}$ is evaluated on is important because degrees of freedom can enter and exit the bulk, and so the modes on a Cauchy slice at time $t_1$ will be different from the modes at time $t_2$. See Figure~\ref{fig:boundaryconditions} for an illustration. 

\begin{figure}
\centering
\includegraphics{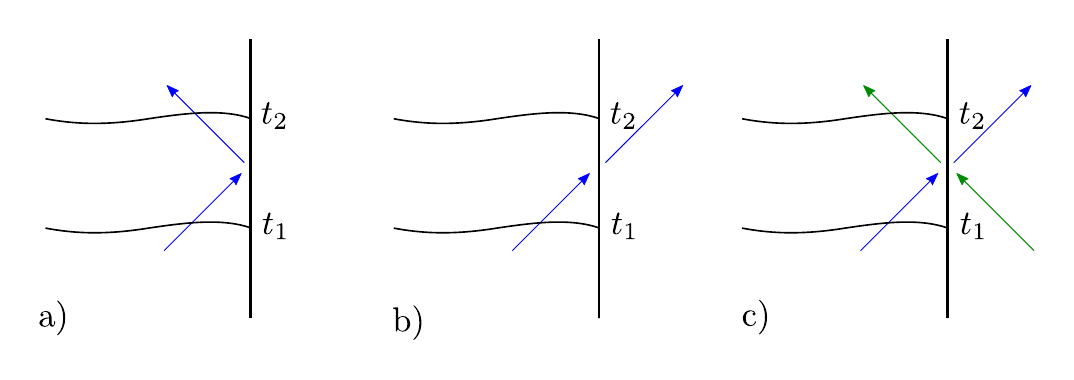}
\caption{We consider three types of boundary conditions: a) reflecting, b) absorbing, and c) open. For absorbing and open boundary conditions, the choice of the Cauchy slice for evaluating the von Neumann entropy is no longer immaterial, as degrees of freedom can enter and exit the bulk.}
\label{fig:boundaryconditions}
\end{figure}     

All of our results in this paper will be valid for arbitrary boundary conditions. We therefore now review the new subtleties and considerations that appear when the boundary conditions are not fixed.

First, with \textit{non-reflecting} boundary conditions, comparison of the generalized entropy of a QES ${\cal X}$ with that of a surface evolved further from it along $k^{a}$, necessitates use of the same boundary Cauchy slice. That is, to properly compare the generalized entropy of different surfaces we need to consider the same configuration of modes in the bulk when computing the von Neumann entropy, and so we must work with a fixed boundary Cauchy slice when the bulk evolution is not unitary. 

To be explicit: with reflecting boundary conditions, the generalized entropy of a surface $\sigma_{R}$ homologous to $R$ may be evaluated along any Cauchy slice containing $\partial R$. This is because the bulk entropy is \emph{independent} of the boundary slice we choose; bulk degrees of freedom may have either reflected off the boundary or not, depending on the slice we use, but the same degrees of freedom will still contribute to the bulk entropy. The generalized entropy only depends on the boundary domain of dependence $D[R]$ of the region $R$, as we would hope given that the reduced states on any two boundary regions with the same domain of dependence are related by a unitary transformation in the CFT. In fact, we can even define bulk entropies for surfaces that are timelike separated from the boundary region $R$, simply by using a boundary Cauchy slice that does not contain $R$. 

By contrast, for general coupled boundary conditions, it does not make sense to compare the bulk entropy (or generalized entropy) of surfaces that are timelike separated from the boundary region $R$, since we cannot use a boundary Cauchy slice that does not contain $R$ without changing the bulk degrees of freedom included in the generalized entropy.

For \emph{absorbing} boundary conditions, where information can escape the spacetime but no new information can enter, this situation is somewhat ameliorated. Specifically, by absorbing boundary conditions, we mean that the bulk evolution from the original Cauchy slice $C$ to a new Cauchy slice $C'$ whose boundary is in the future of the boundary of $C$ is given by a \emph{fixed} quantum channel (independent of the state of any other quantum system).\footnote{For the transparent boundary conditions considered in~\cite{AEMM,AMMZ}, this is not strictly true even when the ingoing modes are fixed to be in the vacuum state, since the vacuum is entangled. However, in practice, we can treat these boundary conditions as purely absorbing for most purposes. For example, we can use the techniques from Section \ref{sec:nonholographic} to consider quantum extremal surfaces for the boundary region $R$ -- plus all ingoing modes that will enter the spacetimes within the future causal diamond of $R$. This is sufficient to make the forwards time evolution to any future slice of $D[R]$ be a fixed quantum channel, and will generally only have a very small effect on the location of the QES. However, see~\cite{AMM} for an example where not including future ingoing (in this case thermal) modes is able to push the extremal surface from exactly on the event horizon to outside the event horizon.} If this is true for \emph{all} forwards boundary evolutions, the system is said to obey a Markovian master equation.

In this case, we can consider surfaces that are in the future of the boundary region $R$, so long as they spacelike separated from some future boundary region $R' \subset D^{+}[R]$ with $\partial R' = \partial R$. The evolution from $R$ to $R'$ will be a quantum channel, which will have a unique purification $E$ (up to isomorphism). We then simply define the bulk von Neumann entropy to be the von Neumann entropy of fields in the bulk region $H_{R'}$ satisfying $\partial H_{R'} = \sigma \cup R'$, together with the purifying system $E$. Since the evolution from $R$ to $R' \cup E$ is unitary, this will be independent of the choice of $R'$, as one would hope. See Figure~\ref{fig:absorbingboundaryconditions}.
This will be useful in discussions of the entropy of representatives of surfaces on other Cauchy slices in later sections.

\begin{figure}
\centering
\includegraphics{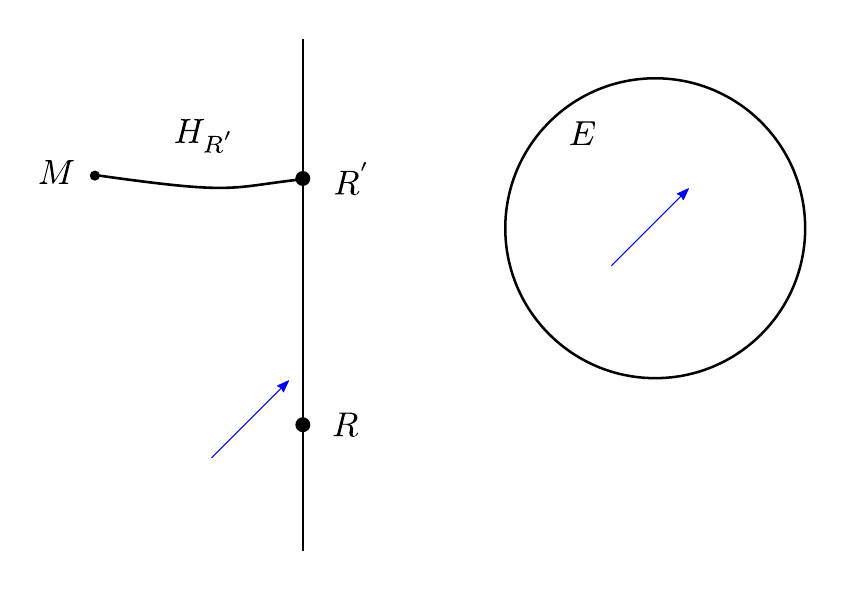}
\caption{Two timelike-separated regions of the boundary are depicted, $R$ and $R'$, with $R' \subset D^+[R]$. The boundary conditions allow modes to pass from the bulk, through the boundary, into the system $E$. The bulk surface $M$ is anchored to $\partial R = \partial R'$ and timelike related to $R$ but spacelike related to $R'$. While $\sigma$ does not lie on any Cauchy slice anchored to $R$, we can nevertheless define an entropy of $\sigma$ with respect to $R$ as the entropy of $H_{R'}\cup E$.}
\label{fig:absorbingboundaryconditions}
\end{figure}   

\section{Quantum Maximin}\label{sec:3}

We begin with a review of the classical maximin prescription of~\cite{Wal12}. For comments on some mathematical subtleties in the classical maximin construction see~\cite{EngHarTA}.

\subsection{Classical Maximin}


Consider a boundary subregion $R$ with boundary $\partial R$. The maximin surface $M_R$ associated to $R$ is defined as follows. Consider all Cauchy slices $C$ such that $\partial R \in C$. On each Cauchy slice $C$ the claim is that there exists a (globally) minimal area codimension-two surface $\text{min}[C,R]$ homologous to $R$. A maximin surface is defined as a surface with (globally) maximal area from the set of minimum area surfaces

\begin{equation}
    M_R = \max_{\mathrm{over\ all \ }C}~ \min_{\mathrm{on \ }C}[C,R]~.
\end{equation}

If multiple such surfaces exist, we say that they are all maximin surfaces. We denote the Cauchy slice on which $M_R$ is minimal as $C[M_R]$. Formally, we also restrict ourselves to stable maximin surfaces as defined in \cite{MarWal19}. A maximin surface $M_R$ is defined as stable if under any (boundary vanishing) infinitesimal deformation of $C[M_R]$ to $C^{'}$ there exists a locally minimal area surface $\sigma$ on $C^{'}$, within a tubular neighborhood of $M_R$, such that $\text{Area}[\sigma] < \text{Area}[M_{R}]$. For now we will take the existence of at least one stable classical maximin surface as a given, returning to the question of existence in Section~\ref{sec:existence}. 

Once existence is established, Wall argues that a maximin surface $M_R$ is an extremal surface of minimal area, or HRT surface, for the region $R$. The basic intuition is that $M_R$ is extremal under variations on the Cauchy slice on which it is minimal, and extremal under variation of the Cauchy slice. Therefore, by the linearity of first order perturbations, $M_R$ is an extremal surface.

To show that it is a \emph{minimal area} extremal surface (or HRT surface), Wall introduces the notion of a ``representative'' of a surface on some Cauchy slice. Let $\sigma$ be some surface in our spacetime, and let $C$ be a Cauchy slice not containing $\sigma$. We define the representatives of $\sigma$ on $C$ as:
\be
\widetilde{\sigma}[C] \equiv \partial J[\sigma]\cap C~,
\ee
where $J[\sigma]\equiv J^{+}[\sigma]\cup J^{-}[\sigma]$. This defines up to two nonempty representatives of $\sigma$; it is possible for the number of nontrivial representatives to be smaller than two, since $\partial J[\sigma]$ can hit the boundary before intersecting $C$, reducing the number of representatives. If $X_{R}$ is homologous to a region with an empty boundary, it can have no nontrivial representatives; if $X_{R}$ is homologous to a region with a boundary, then its representatives must also be homologous to that region and are therefore nonempty. In the case that two representatives exist we take either one of them as the representative and define it as $\widetilde{\sigma}$. See Figure~\ref{fig:representative}. 

\begin{figure}
\centering
\includegraphics{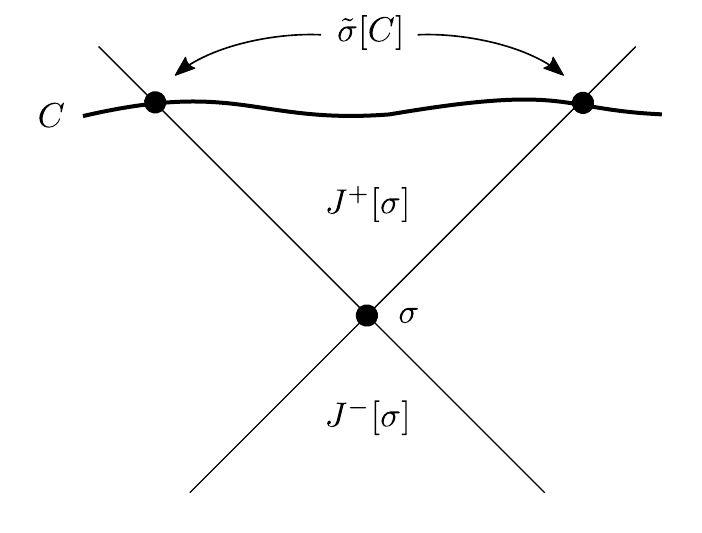}
\caption{We define the representatives of a surface $\sigma$ on Cauchy slice $C$ as the intersection of $C$ with the boundary of $J[\sigma]\equiv J^{+}[\sigma]\cup J^{-}[\sigma]$. We display the case where there are two representatives of $\sigma$ on $C$.}
\label{fig:representative}
\end{figure}     

If $\sigma$ happens to be an extremal surface $X_{R}$, then -- assuming the NCC -- the focusing theorem implies that the area of cross-sections of $\partial J[X_{R}]$ decreases with evolution away from $X_{R}$. 

This is used in Wall's argument in the following way: suppose that $X_{R}$ is an extremal surface; by earlier arguments so is $M_{R}$, so we need to determine which is minimal. To do this, we find the representative of $X_{R}$ on the Cauchy slice $C$ on which $M_{R}$ is minimal (this exists by definition of the maximin). Because $X_{R}$ and $M_{R}$ are both homologous to $R$, the representative must exist. Then we find:
\begin{equation}
\text{Area}[X_R] \geq \text{Area}[\widetilde{X_{R}}] \geq 
\text{Area}[M_R]~,
\end{equation}
where the first inequality follows from focusing, and the second follows from minimality on a Cauchy slice. This implies that $M_{R}$ either has strictly less area than $X_{R}$, in which case it is the minimal area extremal surface and thus the HRT surface, or it has the same area, in which case either surface is admissible as the HRT surface. This establishes the desired result, that $M_R$ is HRT. 

As is evident from the above construction, much of the power of maximin comes from minimality of $M_R$ on a Cauchy slice. Nesting, SSA, and MMI for the covariant holographic entanglement entropy formula are all consequences of this property. Essentially, this permits the reduction of the time-dependent problem to a single time slice, at which point arguments that are used for the static Ryu-Takayanagi formula suffice. This concludes our review of the classical maximin prescription.

\subsection{Defining Quantum Maximin}
\label{sec:QuantumMaximinDefinitions}

We begin with a formal definition of the quantum maximin surface of a boundary subregion, as first described in the introduction. Let $R$ be a spacelike boundary subregion, and assume that $\partial R$ is boundary Cauchy-splitting. 

A quantum maximin surface $M_R$ is obtained  by the following maximinimization procedure: for every Cauchy slice containing $R$,\footnote{We use the term ``containing'' here loosely, meaning that the boundary of the Cauchy slice contains $R$, or equivalently the Cauchy slice of the conformal completion contains $R$.} we find the minimal $S_{\mathrm{gen}}$ surface homologous to $R$. We then look for the maximum $S_{\mathrm{gen}}$ surface among all of these minima. This we term the quantum maximin surface. We will argue that the quantum maximin surface is identical to the surface of the EW proposal, that is, that
\begin{align}
S_{\mathrm{vN}}[\rho_{R}] = \max_{C} \min_{\sigma_R \in C} S_{\text{gen}}[\sigma_R]
\end{align}
where the maximization is over Cauchy slices $C$ containing the boundary region $R$, the minimization is over surfaces $\sigma_R \in C$ that are homologous to the boundary region $R$, and $S_{\text{gen}}[M_R]$ is defined by equation~\ref{eq:Sgen}. 

We shall also use notation as above where $\text{min}[C,R]$ denotes the quantum minimal surface homologous to $R$ on $C$, and $C[M_R ]$ denotes the Cauchy slice on which the quantum maximin surface of $R$ is minimal.

The requirement that the Cauchy slice contain the entire boundary region $R$ is different from the original definition of a classical maximin surface, where the Cauchy slice only had to contain $\partial R$. However it was shown in \cite{MarWal19} that even fixing the entire intersection of the Cauchy slice with the boundary (which they called `restricted maximin') did not affect the surface. We include it here because it is necessary to make the generalized entropy well defined when the boundary conditions are nontrivial, as discussed in Section \ref{sec:boundaryconditions}.

An additional restriction which was imposed in the original formulation of maximin surfaces is ``stability''. Intuitively speaking, it is intended to capture the notion that the maximin surface is not accidental: i.e., that it is robust against small perturbations (of the Cauchy slice). To implement this,~\cite{MarWal19} gave a modification of the original formulation of stability in~\cite{Wal12}. In this definition, a maximin was said to be stable if any small perturbation of the Cauchy slice on which the maximin is minimal results in a Cauchy slice with a local minimum which is nearby the maximin (in spacetime) and has smaller area. This is weaker, but conceptually close to, requiring that the classical maximin be a local maximum over local minima. A type of maximin which can be ruled out by this stability criterion is one which is null-separated to itself or which lives on a null Cauchy slice.
See Figure~\ref{fig:UnstableMaximin} for an illustration of both types.
Due to the poorly-understood nature of entropy on surfaces with null boundaries (see~\cite{BouCas14a, BouCas14b} for special cases which are unfortunately not applicable to our setup), we will slightly strengthen the definition of stability for quantum maximin; note that it is possible in certain cases to relax it to an exact quantum generalization of the criterion of~\cite{MarWal19}:

\begin{figure}
\centering
\includegraphics{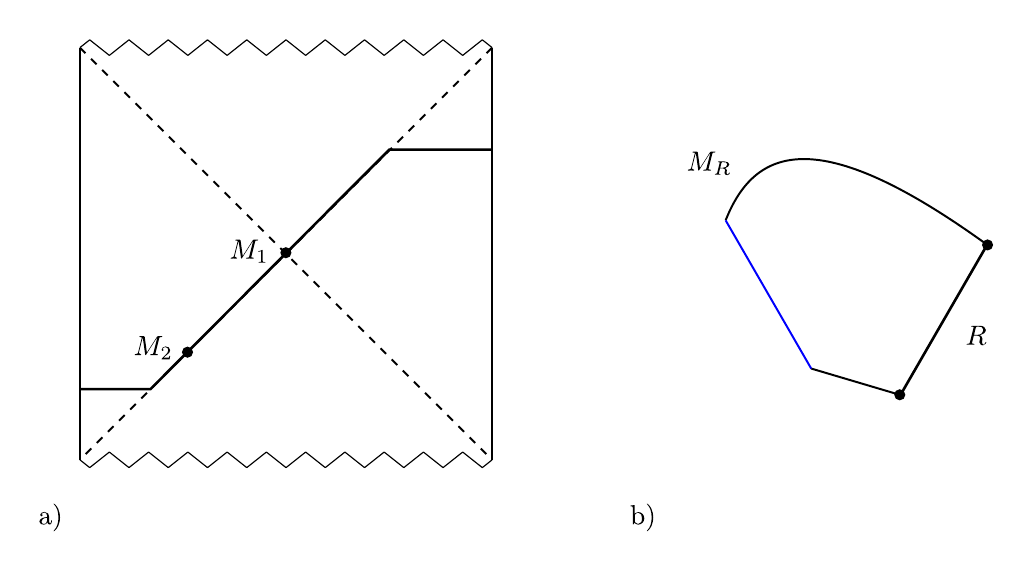}
\caption{The stability requirement rules out certain maximin surfaces. In a) we take our region $R$ to be one entire boundary in the eternal Schwarzschild-AdS geometry. All surfaces on the horizon have the same area, but only $M_1$ at the bifurcate horizon is stable. In b) We have a maximin surface with a null segment (blue), which is not stable.}
\label{fig:UnstableMaximin}
\end{figure}

\begin{defn}
A quantum maximin surface $M_R$ is said to be stable if (1) it is acausal, and (2) any variation of $C[M_R ]$ supported on a tubular neighborhood of $M_R$ has a locally minimal surface $\sigma$ homologous to $R$ with $S_{\mathrm{gen}} [\sigma]<S_{\mathrm{gen}} [M_R ]$.
\end{defn}

\noindent This assumption will be critical to the proof that quantum maximin surfaces are quantum extremal.  It is possible for more than one stable quantum maximin to exist, in which case either might be the quantum maximin. This is similar to the degeneracy in the classical case, where for non-generic configurations of boundary subregions, it is possible for two extremal surfaces to be simultaneously minimal. A small perturbation of the boundary region normally breaks the degeneracy between the two surfaces. Here the situation is somewhat worse: even if one quantum maximin surface has smaller $S_{\mathrm{gen}}$ than another, but the difference is $\mathcal{O}(1)$, it is not possible to determine which surface is the dominant contribution in the EW prescription since the surfaces themselves are only defined up to $\mathcal{O}(1)$ fuzziness. This means that small perturbations may not immediately break the degeneracy. Since the resolution of this lies in transplanckian physics, we will confine ourselves to boundary subregions that do not have this type of degeneracy. 

The proof that quantum maximin surfaces are the dominant quantum extremal surface homologous to a given boundary region requires the notion of a quantum representative of a surface. Let $\sigma$ be a surface, which is by definition Cauchy-splitting; let us denote one side of it on some Cauchy slice $\mathrm{In}[\sigma]$. For \textit{any} Cauchy slice $C$ we define the representatives of $\sigma$ on $C$ as
\begin{equation}
    \widetilde{\sigma}[C] \equiv \partial J[\sigma]\cap C,
\end{equation}
where $J[\sigma]\equiv J^{+}[\sigma]\cup J^{-}[\sigma]$. See Figure~\ref{fig:representative}. This defines either zero, one, or two representatives because $\partial J[\sigma]$ can hit the boundary and terminate, reducing the number of representatives. In the case where the initial surface is a quantum extremal surface ${\cal X}_R$ at least one representative will exist on a Cauchy slice that contains $R$. In the case that two representatives exist we take either one as the representative and denote it by $\widetilde{\sigma}$. The representative is Cauchy splitting, and we take $\mathrm{In}[\widetilde{\sigma}]$ to lie on the same 'side' as $\mathrm{In}[\sigma]$. 

If the bulk evolution is unitary then the von Neumann entropy of $\mathrm{In}[\sigma]$ can be evaluated on any unitarily equivalent Cauchy slice. If $\sigma$ is a quantum extremal surface ${\cal X}_R$ then by applying the quantum focusing conjecture we find that the representative $\widetilde{{\cal X}_R}$ will satisfy $S_{\mathrm{gen}}[\widetilde{{\cal X}_R}] \leq S_{\mathrm{gen}}[{\cal X}_R]$. 

Since by definition the quantum extremal surface must be anchored to the boundary subregion $R$, the above argument goes through without modification for representatives of quantum extremal surfaces on Cauchy slices that contain $R$ in the case of absorbing and coupled boundary conditions.

\subsection{Existence} \label{sec:existence}
 There are several questions that are naturally raised by this definition. Most pressingly, we want to know whether such a surface even has to exist at all. Even classically, this question is exceedingly subtle, and answering it with complete rigor would require mathematical tools far beyond the scope of this paper. Once quantum corrections are included, such rigor would be impossible. That said, in this section we give arguments, at a physics level of rigor, that at least one quantum maximin surface should always exist.

Let us first focus on the question of whether a minimal generalized entropy surface should exist within any fixed Cauchy slice $C$. 

Classically, as argued in \cite{Wal12}, we can choose a topology with respect to which the space of surfaces within $C$ satisfying the homology constraint is compact, and with respect to which area is a lower semicontinuous function. This implies that a minimal area `surface' exists in the completion of the space of surfaces in $C$ (with respect to the appropriate topology). Of course, this argument gives no guarantee that the resulting `surface' (or, more formally, limit of surfaces) is well behaved.

More powerful arguments require the tools of geometric measure theory. The idea is then to  argue (at least for spacelike Cauchy slices) that minimal area surfaces exist as varifolds, which are nonsingular submanifolds as long as the spacetime has dimension at most seven~\cite{bombieri2011geometric}. 

Of course, to define a finite area for a surface, we need to regulate the bulk theory by cutting off the spacetime at some finite bulk radius. Even for minimal area surfaces, we are unaware of a fully mathematically rigorous proof that the minimal area surface is well defined in the limit where the boundary cut-off is taken to infinity.\footnote{However, see~\cite{Sor19} for an argument that the minimal area extremal surface is cut-off independent in this limit.}

Similarly, we do not believe that existence has been demonstrated for entropy functionals that include higher-curvature corrections, even though the functional is still well defined for any given surface in the semiclassical regime (e.g. the derivation in~\cite{Don13} only applies to surfaces with nice curvatures, and probably breaks down for surfaces outside this semiclassical regime). Once the bulk von Neumann entropy term is included in the generalized entropy, even the generalized entropy functional cannot be rigorously defined;  so it is pointless to hope to make any mathematically rigorous argument.  Nonetheless, since $S_{\mathrm{gen}}$ is a well-behaved functional that has had great success in quantum-correcting the area, we still expect that a similar story should hold and that a minimal generalized entropy surface should both exist and be reasonably well behaved. 

In particular, we find it plausible that oscillations in the surface at scales much smaller than the bulk field theory length scale,\footnote{By this we mean the smallest length scale at which bulk excitations exist in the bulk effective field theory.} but larger than the Planck and string length scales, should always increase the generalized entropy.
The rationale for this is that the increase in area from the oscillations should dominate over all other effects. 

If the surface oscillates at transplanckian scales, higher curvature corrections can become larger than the area term. In fact, as discussed in Appendix \ref{appendix}, if such oscillations exist, the sum over higher curvature corrections will generally be divergent, as terms with higher derivatives of the extrinsic curvature will give increasingly large contributions. 

It is unclear how, or even whether, generalized entropy can be defined for such surfaces, since the classical geometry should not be well defined at transplanckian scales. It would therefore maybe be natural to restrict our minimization to only consider surfaces that do \emph{not} oscillate at transplanckian scales. However, we expect that, \emph{if} generalized entropy is indeed well defined for such surfaces, they should never have minimal generalized entropy.
For this reason, we shall not by-hand restrict the surfaces over which we minimize; instead we shall just not worry about such transplanckian oscillations.

The maximization of the minimal generalized entropy over the space of Cauchy slices follows a similar story.
As argued in~\cite{Wal12}, minimal area is an upper semicontinuous function of the Cauchy slice with respect to a topology where the space of Cauchy slices is compact. Hence a maximum exists in the completion of the space of Cauchy slices with respect to this topology. Note that even in classical geometries, however, there are known examples where this maximum fails to exist, since the maximum need not live on an actual Cauchy slice -- only on the limit of one (see~\cite{FisMar14} for some examples). 

For generalized entropy, we again cannot make such precise statements. However, we know that oscillations in the Cauchy slice at small scales will decrease the area of the minimal generalized entropy surface. If the oscillations are at scales that are small compared to the bulk field theory scale, but much larger than the Planck and string scales, this change in area will again be the dominant effect. Transplanckian oscillations will make the higher curvature expansion diverge, and so it is natural to disallow them when we do the maximization. However, it seems reasonable to expect that, if well defined, they should still decrease the minimal generalized entropy and so will not effect the definition of the maximin surface.
So like before, we do not by-hand restrict the Cauchy surfaces over which we maximize. (Note that since we expect that making a spacelike non-acausal surface acausal will only result in an increase in both area and entropy, we also expect that whenever a quantum maximin surface does exist, at least one stable quantum maximin exists as well.) 

Having hopefully convinced the reader that quantum effects most likely do not constitute an obstacle to existence of the quantum maximin, we briefly comment on obstacles to existences which do not result from quantum effects, as alluded to above. It is now understood that \textit{classical} maximin surfaces do not always exist (e.g. in the presence of a dS boundary behind the black hole horizon, as in~\cite{FisMar14}, where the HRT surface had to be complex). Indeed, classical maximin surfaces have only been argued to exist in horizonless spacetimes or spacetimes with certain types of singularities (e.g. Kasner)~\cite{Wal12, MarWal19}. It is thus reasonable to expect that quantum maximin surfaces likely suffer from the same problem. Interestingly, it is in principle possible that a quantum maximin surface could exist where no classical one does (in parallel with the nucleation of a new QES in~\cite{Pen19, AEMM}). We leave an investigation of this possibility to future work. 

\subsection{Equivalence to the Minimal Entropy Quantum Extremal Surface}
\label{sec:MaximinEquivalentQES}

In this section, we show that a stable quantum maximin surface $M_{R}$ is the quantum extremal surface of minimal $S_{\mathrm{gen}}$, appropriately regulated, homologous to $R$. Our derivation is somewhat similar to that of~\cite{Wal12} for the classical maximin, though our assumption of stability is stronger.

We will first prove that $M_R$ is quantum extremal.  When the quantum maximin is unique, we do not need the full power of stability to prove extremality:

\begin{thm} Let $M_{R}$ be the unique quantum maximin surface of $R$, and assume only condition (1) of stability. Then $M_R$ is quantum extremal.\end{thm}

\begin{proof}
Let us first consider the case where the Cauchy slice on which $M_R$ is minimal, $C[M_{R}]$, has continuous first derivatives in all directions in a tubular neighborhood of $M_R$. Let $r^{a}$ be the normal to $M_R$ on $C[M_R ]$. By definition, $M_R$ is a global minimum of $S_{\mathrm{gen}}$ on $C[M_R ]$. As discussed in Section~\ref{sec:existence}, we expect that the minimal surface on a Cauchy slice will have a continuous first derivative, so it can be continuously varied infinitesimally along a vector field $q^{a}$. Such small variations will result in surfaces with continuous first derivative and an $S_{\mathrm{gen}}$ which is either larger than or identical to $S_{\mathrm{gen}}[M_R ]$. Thus to first order in variations along $q^{a}$, the change in $S_{\mathrm{gen}}[M_R ]$ is either zero or positive. In terms of the quantum mean curvature ${\cal K}_{a}$~\eqref{eqn:qmc}, this immediately implies ${\cal K}_{a}q^{a}\geq 0$. Since the statement is true for $q^{a}\rightarrow -q^{a}$, and since any such variation can be decomposed $q^{a}=s^{a}+r^{a}$, where ${\cal K}_{a}s^{a}=0$ and $r^{a}$ is normal to $M_R$ on $C[M_R ]$, we immediately find ${\cal K}^{a}r_{a}=0$. Because by stability, $M_R$ is not null-separated to itself, there exists a tubular neighborhood $U$ of $M_R$ on $C[M_R ]$ such that no two points on $U$ are null-separated. We can therefore vary $C[M_R ]$ by varying $U$ infinitesimally along the future-directed timelike vector $t^{a}$ normal to $C[M_R ]$; call this deformed slice $C'$. Let $\text{min}[C',R]$ be the minimal $S_{\mathrm{gen}}$ surface on the deformed slice. By maximality, $S_{\mathrm{gen}}[M_R ]> S_{\mathrm{gen}}[ \text{min}[C',R]]$. Because $M_R$ is unique, there exists a sufficiently small choice of $U$ such that $\text{min}[C',R]$ lies in a small tubular spacetime neighborhood of $M_R$; that is min$[C',R]$ may be obtained by an infinitesimal variation of $M_R$ along some vector field $m^{a}$ (which we may assume by above has no components along $M_R$). By maximality, we find that ${\cal K}_{a}m^{a}\leq 0$. We may repeat this argument with a past-directed timelike vector, obtaining ${\cal K}_{a}n^{a}\leq0$ for some $n^{a}$ with no components along $M_R$. Because $M_R$ is codimension-two, its normal bundle has only two independent sections: we may therefore decompose $m^{a}=a r^{a}+b t^{a}$ and $n^{a} = \alpha r_{a} + \beta t_{a}$  where $b>0$ and $\beta<0$. It then immediately follows that ${\cal K}_{a}t^{a}=0$, and therefore ${\cal K}_{a}=0$, and $M_R$ is extremal. 

We now consider the situation where $C[M_R ]$ is not smooth, although strictly speaking, we smear over any Planckian neighborhoods, so this treatment is not entirely necessary; we include it for completeness. This portion of the proof will use Definition 13 of~\cite{Wal12} for tangent vectors to a Cauchy slice which has a discontinuous first derivative. The argument is nearly identical; variations of the surface along the outwards direction tangent to $C[M_R ]$, $r^{a}$, yield ${\cal K}_{a}r^{a}\geq 0$. Variations of $M_R$ along the inwards direction tangent to $C[M_R ]$, which we will call $p^{a}$ and no longer assume that $p^{a}=-q^{a}$, also yield ${\cal K}_{a}p^{a}\geq 0$. Similarly, we still obtain ${\cal K}_{a}m^{a}\leq0$ and ${\cal K}_{a}n^{a}\leq0$, although we no longer assume that $m^{a}$ and $n^{a}$ are obtained by varying along $t^{a}$ and $-t^{a}$. However, we may now easily obtain two vector fields on $M_R$ $\{w^{a},y^{a}\}$ such that ${\cal K}_{a}w^{a}=0={\cal K}_{a}y^{a}$. Then either $w^{a}$ and $y^{a}$ are linearly independent and $M_R$ is extremal, or $w^{a}$ and $y^{a}$ are not linearly independent, in which case two of the deformation vectors are diametrically opposed. But then we know that along that direction, say $r^{a}$, ${\cal K}_{a}r^{a}=0$; we can then decompose $m^{a}$ and $n^{a}$ in terms of $r^{a}$ and $t^{a}$ as above, which again shows that $M_R$ is extremal.
\end{proof}

In the event that the quantum maximin is degenerate, we prove that the stable one(s) is (are) extremal:
\begin{thm} Let $M_R$ be a stable quantum maximin surface of $R$. Then $M_R$ is extremal.
\label{thm:stableisextremal_true}
\end{thm}

Note that since, in the regime that the area variation contributes at leading order to the quantum mean curvature while the entropy variation contributes only at subleading order, condition (1) of stability follows from condition (2), in such a regime it is sufficient to assume condition (2) only. 

\begin{proof}
The proof is almost identical, with the exception of the component that relies on the uniqueness of the maximin to argue that min$[C',R]$ can be obtained by a small deformation of $M_R$ along a direction which does not live in the tangent bundle of $C[M_R ]$. With assumption (2) of stability, however, this becomes unnecessary: maximality guarantees that $S_{\mathrm{gen}}[M_R ] \geq S_{\mathrm{gen}}[\sigma]$, which gives the requisite sign for ${\cal K}_{a}m^{a}$ and ${\cal K}_{a}n^{a}$. 
\end{proof}

We now proceed to the desired result: a stable quantum maximin is the quantum extremal surface of minimal $S_{\mathrm{gen}}$ homologous to a boundary subregion. 

\begin{thm} Let $M_{R}$ be a stable quantum maximin of $R$. Then $M_{R}$ has minimal $S_{\mathrm{gen}}$ over all quantum extremal surfaces homologous to $R$. \label{thm:stableisextremal}
\end{thm}

\begin{proof}
Let ${\cal X}_{R}$ be an $S_{\mathrm{gen}}$-minimizing quantum extremal surface homologous to $R$. We would like to show that $S_{\mathrm{gen}}[{\cal X}_{R}]=S_{\mathrm{gen}}[M_{R}]$. If ${\cal X}_{R}\subset C[M_R ]$, then by minimality of the quantum maximin, $S_{\mathrm{gen}}[M_{R}]\leq S_{\mathrm{gen}}[{\cal X}_{R}]$; however, since $M_{R}$ is quantum extremal by Theorem~\ref{thm:stableisextremal_true}, this means that ${\cal X}_{R}$ cannot be the minimal quantum extremal surface homologous to $R$ unless $S_{\mathrm{gen}}[{\cal X}_{R}]=S_{\mathrm{gen}}[M_R]$, which proves the desired result when ${\cal X}_{R}\subset C[M_{R}]$. 

We now consider the case where ${\cal X}_{R}$ does not lie on $C[M_R]$. We will prove that it has a representative on $C[M_R]$ that is homologous to $R$. This we do in two cases: ${\cal X}_{R} \cap C[M_R]=\partial R$ and $\partial R \subsetneq {\cal X}_{R} \cap C[M_R]$. In the former case, ${\cal X}_{R}$ lies on one side of $C[M_R]$. We consider firing a null congruence $N$ from ${\cal X}_{R}$ towards $C[M_R]$. Then because ${\cal X}_{R}$ lies outside of the domain of influence of $R$, $I^{+}_{bulk}[R]\cup I^{-}_{bulk}[R]$, so must $N$. Thus there exists a representative $\widetilde{{\cal X}_{R}}$ of ${\cal X}_{R}$ on $C[M_R]$. Because  $\widetilde{{\cal X}_{R}}$ is homologous to ${\cal X}_{R}$ and ${\cal X}_{R}$ is homologous to $R$,  $\widetilde{{\cal X}_{R}}$ is also homologous to $R$. 

In the latter case,\footnote{This case was not considered in~\cite{Wal12}, but we see no reason why it can be obviously excluded.} ${\cal X}_{R}$ intersects $C[M_R]$ in the bulk interior. By~\cite{BouEng15b}, $\partial D^{+}[H_R]\cup \partial J^{-}[H_R]$, where $H_R$ is the homology hypersurface of ${\cal X}_{R}$ (and the opposite combination as well) forms a single  null congruence of which ${\cal X}_{R}$ is one slice. Any slice of this null congruence is homologous to ${\cal X}_{R}$, and therefore also to $R$. By the same argument above, the representative  $\widetilde{{\cal X}_{R}}$ exists, and is thus also homologous to $R$.

As shown in Sec.~\ref{sec:QuantumMaximinDefinitions}, $S_{\mathrm{gen}}$ of a representative is smaller than that of ${\cal X}_{R}$. We thus find:

\be
S_{\mathrm{gen}}[{\cal X}_{R}]\geq S_{\mathrm{gen}}[\widetilde{{\cal X}_{R}}] \geq S_{\mathrm{gen}}[M_R]~.
\ee

We thus immediately find that $M_R$ is the minimal $S_{\mathrm{gen}}$ quantum extremal surface homologous to $R$. 

\end{proof}

\section{Applications}\label{sec:applications}

We now discuss a few important applications of the equivalence between quantum maximin surfaces and QESs. 
The first is entanglement wedge nesting, which heuristically is the property that a smaller boundary region has a smaller entanglement wedge. 
While interesting in its own right, this feature is vital for proving other properties of quantum maximin surfaces. 
In particular, we use it to prove that the EW prescription always obeys the strong subadditivity inequality -- a crucial consistency check of any proposal for calculating entropies.

As a warm up, we first review another important property of the EW prescription: the entanglement wedge of a boundary region $R$ always contains the causal wedge of $R$, defined as the intersection of the bulk past and future of the boundary domain of dependence of $R$. This property does not require the power of quantum maximin and was proved using just the Generalized Second Law in \cite{EngWal14}. Our main purpose in reviewing it here is to comment on its application to spacetimes with non-reflecting boundary conditions, as discussed in Section \ref{sec:boundaryconditions}.

\paragraph{Entanglement Wedges Contain Causal Wedges}
\begin{thm}\label{thm:causal}
	Let $R$ be a boundary region, with domain of dependence $D[R]$, and homology hypersurface $H_R$. Then the causal wedge $W_C[R] = J^-[D[R]] \cap J^+[D[R]]$  is entirely contained in the entanglement wedge $D[H_R]$, where $J^-[D[R]]$ (resp. $J^+[D[R]]$) is the bulk past (resp. bulk future) of $D[R]$ viewed as a bulk submanifold. 
\end{thm}
Let us review the proof of this result from \cite{EngWal14}. Suppose, by way of contradiction, that the causal wedge $W_C[R]$ is not contained in the entanglement wedge $D[H_R]$. By continuity, we could then deform $D[R]$ to some smaller boundary spacetime region $D' \subseteq D[R]$ that is not necessarily a domain of dependence such that either the past or future causal boundary $\partial J^\pm[D']$ intersects the quantum extremal surface ${\cal X}_R$ but does not intersect the exterior of $D[H_R]$. Let $M_{D'}$ be the intersection of $\partial J^\pm[D']$ with a spacelike Cauchy slice containing ${\cal X}_R$ and $x$ be some point in the intersection $M_{D'} \cap {\cal X}_R$. By Wall's lemma (lemma~\ref{Wallemma}), we have
\begin{align}
    \Theta[M_{D'},x] \leq \Theta[{\cal X}_R, x] = 0.
\end{align}
However,  the generalized second law (GSL) implies that
\begin{align}
    \Theta[M_{D'},x] \geq 0,
\end{align}
with equality nongeneric.\footnote{This can be more directly derived from the QFC; however, since the QFC is a strictly stronger assumption than the GSL (the QFC implies the GSL), we prefer to use the GSL where it is sufficient.} We have therefore derived a contradiction in the generic case, and, by continuity, we have also proved the result in the nongeneric case.

To extend this result to general boundary conditions, the main point to check is that the generalized second law still applies. In particular, the causal horizon may be timelike separated from the boundary region $R$. As discussed in Section \ref{sec:boundaryconditions}, this makes it impossible to define generalized entropies for arbitrary slices of the causal horizon without knowing anything about the boundary conditions.

Fortunately, as discussed in \cite{Pen19}, the boundary domain of dependence $D[R]$ of a region $R$, evolving using absorbing boundary conditions, is only the \emph{future} causal diamond of $R$, since the backwards time evolution is nondeterministic. Since $D' \subseteq D[R]$, it follows that both the past and future causal horizons $\partial J^\pm[D']$ are either spacelike or future (timelike or lightlike) separated from $R$, and spacelike separated from some slice of $D[R]$. We can therefore use the definition of generalized entropy for absorbing boundary conditions from Section \ref{sec:boundaryconditions} to define generalized second laws for these horizons. Assuming this version of the GSL, we can prove Theorem \ref{thm:causal} for spacetimes with absorbing boundary conditions by exactly the same arguments given above.\footnote{See also~\cite{AMM} for an example of this in AdS$_{2}$.}

What about general coupled boundary conditions? In this case, we cannot define a bulk entropy for any surfaces that are timelike separated from $R$. However, the coupling means that the boundary domain of dependence of $R$ is simply $R$ itself. Consequently, the causal wedge of $R$ is empty,\footnote{Or the causal wedge is $R$ itself, if the asymptotic boundary is included in the bulk spacetime.} and is trivially contained in the entanglement wedge.

\subsection{Nesting}

\begin{thm}\label{thm:Nesting}
	Let $R_1 \subset D[R_2]$ be a boundary region contained inside the domain of dependence $D[R_2]$ of a boundary region $R_2$. Let $M_{R_1}$ and $M_{R_2}$ be their respective quantum maximin surfaces, and $H_{R_1}$ and $H_{R_2}$ be their homology surfaces. Then the domain of dependence of $H_{R_1}$ is contained in that of $H_{R_2}$, with $M_{R_1}$ spacelike from $M_{R_2}$. Furthermore, $M_{R_1}$ and $M_{R_2}$ are minimal on the same time slice. 
\end{thm}

\begin{proof}
    Here we generalize the arguments of~\cite{HeaTak} in the static case and~\cite{EngHarTA} in the covariant case to prove Nesting. For clarity and since~\cite{EngHarTA} is still in preparation, we emphasize where this proof differs from the naive quantum generalization of Wall's proof in~\cite{Wal12}.

        Let $R_1 \subset R_2$ be boundary regions.
        We consider maximinimizing the quantity $\alpha S_\mathrm{gen}[M_1] + \beta S_\mathrm{gen}[M_2]$, for acausal $M_1,M_2$ homologous to $R_1, R_2$ respectively and $\alpha,\beta$ arbitrary positive real numbers (we assume that these surfaces exist based on the intuition discussed in Section \ref{sec:3}).
	Let $H_{1},H_{2}$ be the homology hypersurfaces of $M_1,M_2$ respectively. The surfaces $M_1$, $M_2$ found this way are both minimal $S_\mathrm{gen}$ surfaces defined on the same time slice $C$. We will eventually show that they are in fact the maximin surfaces $M_{R_1}$ and $M_{R_2}$.
	
\begin{figure}
\centering
\includegraphics{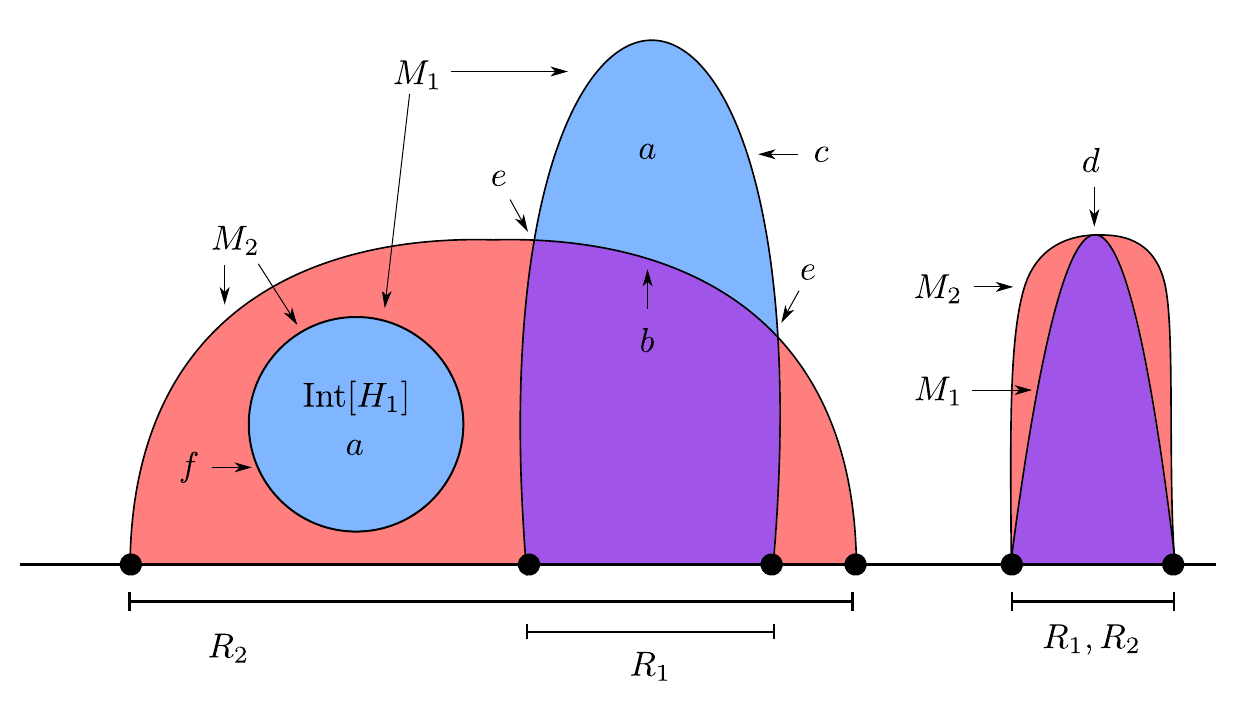}
\caption{A generic setup that violates entanglement wedge nesting. Two boundary regions are depicted, $R_2$ and $R_1 \subset D[R_2]$, with their respective quantum maximin surfaces $M_2$ and $M_1$. All types of points that would violate nesting are labelled with a letter. Points of type $a$ are codimension one subregions, while all other points are subregions of the maximin surfaces. Regions shaded light blue correspond to $\mathrm{Int}[H_1]$ while regions shaded red correspond to $\mathrm{Int}[H_2]$. Regions which are contained by both homology hypersurface, $\mathrm{Int}[H_1] \cap \mathrm{Int}[H_2]$, are shaded purple.}
\label{fig:EWNesting}
\end{figure}
	
We therefore want to show that $H_{1} \subset H_{2}$. We will also need to show that $M_1 \cap M_2$ is a closed and open subset of $M_1$ and $M_2$ (i.e. $M_1$ and $M_2$ only intersect on entire connected components). To do so, there are exactly six types of bulk points that we need to rule out, as enumerated in \cite{EngHarTA}:
			\begin{align}
				a =&\mathrm{Int}[H_1] \cap \mathrm{Ext}[H_2] \label{pt:a}~,\\
				b =& \mathrm{Int}[H_1]\cap M_2  \label{pt:b}~,\\
				c =& \mathrm{Ext}[H_2] \cap M_{1}\label{pt:c}~,\\
				d =& M_1 \cap M_{2} \text{ (same boundary anchors, do not everywhere coincide)}\label{pt:d}~,\\
				e =& M_1 \cap M_{2} \text{ (different boundary anchors)}\label{pt:e}~,\\
				f =& M_1 \cap M_{2} \text{ (floating components, opposite orientation)}\label{pt:f}~.
			\end{align}
Define $\mathrm{Int}[h]$, where $h$ is a manifold with boundary, to be $h - \partial h$. 
	Similarly, $\mathrm{Ext}[h] \equiv C - h$.
         
Define $H_1' = \mathrm{Closure} (\mathrm{Int}(H_1) \cap \mathrm{Int}(H_2))$ and $H_2' = H_1 \cup H_2$ and define surfaces $M_1',M_2'$ by $\partial H_1' = R_1 \cup M_1'$ and $\partial H_2' = R_2 \cup M_2'$. See Figure~\ref{fig:EWNesting}.
By strong subadditivity of bulk von Neumann entropy, 
        \begin{align}
            S_{\mathrm{vN}}[H_1'] + S_{\mathrm{vN}}[H_2'] \le S_{\mathrm{vN}}[H_1] + S_{\mathrm{vN}}[H_2]~.
        \end{align}
        Furthermore, 
	\begin{align}\label{eqn:nestingareas}
            A[M_1'] + A[M_2'] \le A[M_1] + A[M_2]~,
        \end{align}
        where equality holds unless there are points of type $f$, in which case the inequality is strict.
Note now that the surfaces $M_1',M_2'$ will in general have corners, which we must treat carefully because they have ill-defined extrinsic curvatures and therefore poorly defined higher derivative corrections to the geometric part of the generalized entropy. 
		To handle this, we define $M_1',M_2'$ with these corners ``smoothed out'' at a scale large relative to the Planck length and small compared to the bulk field theory scale.  
		Smoothing the corners at a scale small relative to the bulk field theory scale means that the von Neumann entropy will not change appreciably. 
		All changes to the generalized entropy will come from the effect on the geometric part of the generalized entropy.
		Furthermore, smoothing will decrease the area term, while the higher derivative terms will become subdominant to the area term, because we are smoothing at a scale that is large relative to the Planck length and so the semiclassical expansion is valid. \label{remark:corners}
 Therefore, if there are points of type $a$ then
		\begin{align}	
			S_\mathrm{gen}[M_1'] + S_\mathrm{gen}[M_2'] < S_\mathrm{gen}[M_1] + S_\mathrm{gen}[M_2]~,
		\end{align}
		where the inequality is strict because one of two things must be true to have type $a$ points.
		Either there are points of type $f$ or points of type $e$. 
		If there are points of type $f$, then the inequality is strict because Eqn.~(\ref{eqn:nestingareas}) is strict.
		If there are points of type $e$, then the inequality is strict because smoothing the corners strictly reduces the generalized entropy by the previous remark.
 The minimality of $M_1,M_2$ therefore means that there are no points of type $a$ and $H_1 \subseteq H_2$.
		We can also rule out points of type $b,c$ and $f$, because their existence implies the existence of points of type $a$. It also rules out points of type $e$ for the same reason, unless the surfaces are tangent at that point.
		
		We now need to rule out points of type $d$ and points of type $e$ where the surfaces are tangent (i.e. show that $M_1$ and $M_2$ only intersect on entire connected components). 
        Suppose $M_1$ and $M_2$ intersect and are tangent at some point $x$, but do not coincide in an open neighbourhood of $x$ (i.e. they diverge at $x$). Because $H_1 \subseteq H_2$, we can apply Lemma~\ref{SpacelikeLemma} to bound ${\cal K}_2 > {\cal K}_1$ at some point $y$ near $x$, where ${\cal K}_i$ is the quantum mean curvature of surface $M_i$ contracted along the same spacelike orthogonal direction tangent to the time slice.
        Hence either $S_\mathrm{gen}[M_2]$ could be decreased by deforming $M_2$ outwards (making $H_2$ bigger) or $S_\mathrm{gen}[M_1]$ could be increased by deforming $M_1$ inwards (making $H_1$ smaller). 
        Therefore where $M_1$ and $M_2$ have points that coincide, the surfaces match for the entire connected component.
\begin{figure}
\centering
\includegraphics{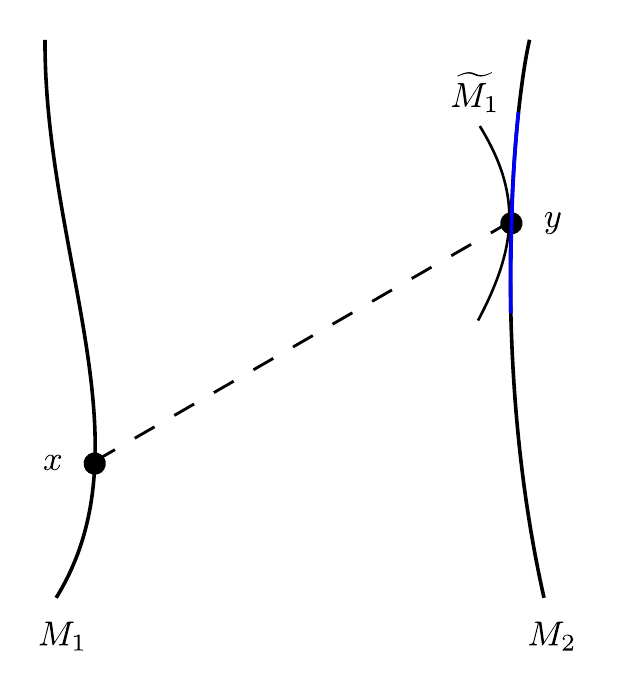}
\caption{We assume that the points $x,y$ on $M_1,M_2$ respectively are null separated (dashed line). Part of the representative $\widetilde{M_1}$ is depicted in blue and is tangent to $M_2$. By applying the touching lemma to the pictured setup we arrive at a contradiction, proving that no two points on $M_1,M_2$ are null separated.}
\label{fig:TouchingLemma}
\end{figure}        

   We now argue that points on $M_1$ cannot be null separated from points on $M_2$.
			For contradiction and without loss of generality, let $x$ be a point on $M_{1}$ null-separated from a point $y$ on $M_{2}$ to its future, i.e. $y \in \dot{J}^+(x)$. 
			Consider shooting a null congruence out from $M_1$ near $x$ towards $y$. Let $\widetilde{M}_1$ be a deformation of $M_1$ near $x$ along this congruence that is tangent to $M_{2}$ at $y$. We can always choose this deformation so that the homology hypersurface $\widetilde{H}_1$ of $\widetilde{M}_1$ in any Cauchy slice is always contained in the domain of dependence $D(H_2)$ of the homology hypersurface $H_2$.
			Because $M_1$ is nowhere timelike related to $M_{2}$, $\widetilde{M}_1$ must ``bend away'' from $M_2$ at $y$ as in Figure~\ref{fig:TouchingLemma}. Since we also have $\widetilde{H}_1 \subseteq D(H_2)$, Wall's lemma then implies
			\begin{align}
				\Theta[\widetilde{M}_1, y] \ge \Theta[M_2, y]~.
			\end{align}
			It must be the case that $\Theta(M_2, y) \ge 0$, or else we could decrease $S_\mathrm{gen}[M_2]$ by deforming it at $y$ to the past along the null generator connecting $x$ and $y$. 
			Moreover, the QFC and quantum null generic condition together imply $\Theta[\widetilde{M}_1, y] < \Theta[M_1, x]$.
			Therefore,
			\begin{align}
				\Theta[M_1, x] > 0~.
			\end{align}
			Consider the quantum-mean curvature ${\cal K}^a$ defined in Eqn.~(\ref{eqn:qmc}),
			\begin{align}
				{\cal K}_a[\sigma] =  4 G_N \hbar \frac{{\cal D} S_\mathrm{gen}[\sigma]}{{\cal D} X^a}~.
			\end{align}
			Let $k^a$ be the null tangent vector to $M_1$ at $x$.
			Then by definition, $k^a {\cal K}_a[M_1,x] = \Theta[M_1, x]$.
			There are two normal vectors to $M_1$ at $x$, and $k^a$ is one of them.
			Since $M_1$ has minimal $S_\mathrm{gen}$ on $C$, there must exist a spacelike tangent vector $v^a$, tangent to $C$, with
			$v^a {\cal K}_a \geq 0$. 
			On the other hand, deforming $M_1$ in a future timelike direction $t^a$ must decrease $S_\mathrm{gen}$, since we can deform $C$ forwards in time at $x$ without affecting $M_2$. Hence $t^a {\cal K}_a \leq 0$. 
			This is a contradiction, since $t^a$ is a positive linear combination of $v^a$ and $k^a$.
	 
	 On components where they do not touch, we can therefore freely vary $C$ in the neighborhood of one surface without affecting the $S_\mathrm{gen}$ of the other surface. 
Where they do touch, they coincide exactly on some connected component $M_C$ (as argued above), and we can take them to continue to coincide on this component for small deformations of $C$.
At this point, we would {\it like} to conclude that $M_1$ and $M_2$ are each quantum extremal, using the arguments in Section~\ref{sec:3}.
However, it is not obvious that the maximinization of $\alpha S_\mathrm{gen}[M_1] + \beta S_\mathrm{gen}[M_2]$ found surfaces that are {\it independently} maximal under deformations of $C$ near $M_C$. 
We now prove that it did.\footnote{Note that the necessity of this part of the proof is a key difference from the classical nesting proof. While the $S_\mathrm{grav}$ term in $S_\mathrm{gen}$ is clearly independently maximal for the two surfaces, things are less obvious for the von Neumann entropy term because it is non-local.} 

The surfaces $M_1$ and $M_2$ can only coincide if 
		\begin{align}
		   v^a {\cal K}_a (M_1, y) = v^a {\cal K}_a (M_2, y) = 0
		\end{align}
		for all $y \in M_C$ and $v^a$ tangent to $C$, since $M_1$ and $M_2$ are separately minimal in $C$. 
		Because $S_\mathrm{grav}$ is local, it contributes equally to both sides, and hence
		\begin{align}
		   v^a \frac{{\cal D}  S_\mathrm{vN}(M_1)}{{\cal D} X^a[y]} = v^a \frac{{\cal D}  S_\mathrm{vN}(M_2)}{{\cal D} X^a[y]}~.
		\end{align}
		Now assume for contradiction that there exists some timelike vector $t^a$ such that
		\begin{align}
		    t^a\frac{{\cal D}  S_\mathrm{vN}(M_1)}{{\cal D} X^a[y]} - t^a\frac{{\cal D}  S_\mathrm{vN}(M_2)}{{\cal D} X^a[y]} =
		    \kappa~,
		\end{align}
		where $\kappa$ is some non-zero real number.
		Then we can consider a spacelike vector $\omega^a = v^a - \mathrm{sign}(\kappa)\epsilon t^a$, where $v^a$ points in the direction away from $H_1$ and $H_2$, and $\epsilon$ is some small positive constant. 
		This would satisfy
		\begin{align}
		   w^a \frac{{\cal D}  S_\mathrm{vN}(M_1)}{{\cal D} X^a[y]} < w^a \frac{{\cal D} S_\mathrm{vN}(M_2)}{{\cal D} X^a[y]}~,
		\end{align}
		which violates strong subadditivity because $H_1 \subseteq H_2$.
		Therefore
		\begin{align}
		{\cal K}_a (M_1, y) = {\cal K}_a (M_2, y)~,
		\end{align}
		for all $y \in M_c$.
		 Because $M_1$ and $M_2$ are acausal we therefore conclude that $M_1$ and $M_{2}$ are each quantum extremal surfaces. 
		 Finally we need to show that $M_1$ and $M_{2}$ are both minimal generalized entropy quantum extremal surfaces and hence maximin surfaces. If they were not, then some other pair ${\cal X}_{1},{\cal X}_{2}$ with less weighted $S_\mathrm{gen}$ would be the minimal quantum extremal surfaces. Then $S_\mathrm{gen}[M_1] > S_\mathrm{gen}[{\cal X}_{1}]$ or $S_\mathrm{gen}[M_{2}] > S_\mathrm{gen}[{\cal X}_{2}]$. But in the former case, the representative $\widetilde{{\cal X}_{1}}$ would have less $S_\mathrm{gen}$ on $C$ than $M_1$, contradicting the minimality of $M_1$. A similar contradiction is reached in the latter case. 
\end{proof}

An immediate corollary follows

 \begin{cor} for any set of disjoint spacelike regions $R_n, n \in 1...N$, all the $M_{R_n}$ are minimal on the same slice $C$. To prove this, construct surfaces $M_n$ by minimizing then maximizing a quantity $Z = \sum_n c_n S_\mathrm{gen}[M_n] $. Because these regions are disjoint, each $M_n$ is spacelike related to any other. The proof that the $M_n$s are the minimal quantum extremal surfaces proceeds in the same way as above.
\end{cor}

\subsection{Strong Subadditivity}

We now use this theorem to prove strong subadditivity. 
There are two features of this proof that are different than the classical one in \cite{Wal12}. 
The first is that we must use {\it bulk} strong subadditivity. 
The second is that we must be especially careful dealing with large extrinsic curvatures in our surfaces, because of their explicit appearance in $S_\mathrm{gen}$. 

\begin{thm}
	Let $R_1,R_2,R_3$ be disjoint boundary regions. Let $M_{i j...k}$ denote the QMM surface for $R_i \cup R_j \cup...\cup R_k$. Then strong subadditivity holds:
	\begin{equation}
		S_\mathrm{gen}[M_{12}] +  S_\mathrm{gen}[M_{23}] \ge S_\mathrm{gen}[M_{123}] + S_\mathrm{gen}[M_{2}]~.
	\end{equation}
\end{thm}

\begin{proof}
	Our strategy is to first find representatives of $M_{12}$ and $M_{23}$ on the same time slice on which $M_{123}$ and $M_2$ lie. Then the minimality of $M_{123}$ and $M_2$, combined with bulk strong subadditivity, will imply the strong subadditivity inequality above. 

 By Theorem \ref{thm:Nesting}, $M_{123}$ and $M_2$ lie on the same time slice $C$. 
 Moreover, $M_{12}$ and $M_{23}$ have homologous representatives $\widetilde{M}_{12}$ and $\widetilde{M}_{23}$ on $C$, and these representatives have smaller $S_\mathrm{gen}$ than $M_{12}$ and $M_{23}$ respectively.
 Define four bulk regions as follows. Let $\partial \widetilde{H}_{12} = R_1 \cup R_2 \cup \widetilde{M}_{12}$ and $\partial \widetilde{H}_{23} = R_2 \cup R_3 \cup \widetilde{M}_{23}$. Let $H'_{2} = \mathrm{Closure}(\mathrm{Int}(\widetilde{H}_{12}) \cap \mathrm{Int}(\widetilde{H}_{23}))$ and $ H'_{123} = \widetilde{H}_{12} \cup \widetilde{H}_{23}$. 
		The boundaries of these latter regions define new bulk surfaces: let $M'_2$ be defined by $\partial H'_2 = R_2 \cup M'_2$. Similarly, let $\partial H'_{123} = R_1 \cup R_2 \cup R_3 \cup M'_{123}$.
		As in remark \ref{remark:corners} of the Nesting proof, we define $M'_2$ and $M'_{123}$ with corners smoothed out. 
We now claim that bulk strong subadditivity implies 
		\begin{equation}
			S_\mathrm{gen}[\widetilde{M}_{12}] + S_\mathrm{gen}[\widetilde{M}_{23}] > S_\mathrm{gen}[M'_{123}] + S_\mathrm{gen}[M'_{2}]~.
		\end{equation}
		To prove this, expand out $S_\mathrm{gen}$ such that the inequality reads 
		\begin{equation}
			\frac{A[\widetilde{M}_{12}]}{4 G_N} +  S_{\mathrm{vN}}[\widetilde{H}_{12}] + \frac{A[\widetilde{M}_{23}]}{4 G_N} + S_{\mathrm{vN}}[\widetilde{H}_{23}] > 
			\frac{A[M'_{123}]}{4 G_N} +  S_{\mathrm{vN}}[H'_{123}] + 
			\frac{A[M'_{2}]}{4 G_N}  + S_{\mathrm{vN}}[H'_{2}]~,
		\end{equation}
		where here we lump all geometric terms together under the label $A/4G_N$.
		As noted above, the corner-smoothing involved in defining $M'_{123}$ and $M'_2$ overall decreases their $A/4G_N$. Therefore the $A/4G_N$ part of the left-hand side is strictly greater than that on the right-hand side. What remains is implied by bulk strong subadditivity. Therefore the entire inequality is satisfied.
 Minimality of $M_{123}$ and $M_2$ on $C$ implies $S_\mathrm{gen}[M'_{123}] + S_\mathrm{gen}[M'_{2}] \ge S_\mathrm{gen}[M_{123}] + S_\mathrm{gen}[M_2]$. Combining this with the above inequalities concludes the proof.

\end{proof}

\subsection{MMI and other inequalities}

The classical maximin construction allowed the proof of inequalities that do not hold for general von Neumann entropies. 
One example is the proof of Monogamy of Mutual Information (MMI)~\cite{HayHeaMal13}, though there are many others (see literature starting with \cite{BaoNez15}). 
We expect in general that these inequalities can be violated at next-to-leading order by QMM surfaces, because there are generally configurations that saturate the inequality at leading order and will not hold at next-to-leading-order, if the entropy inequality is not obeyed by the bulk fields. 

We consider the MMI inequality
\begin{equation}\label{eqn:mmi}
	S_\mathrm{gen}[M_{12}] + S_\mathrm{gen}[M_{13}] + S_\mathrm{gen}[M_{23}] \ge S_\mathrm{gen}[M_{1}] + S_\mathrm{gen}[M_{2}] + S_\mathrm{gen}[M_{3}] + S_\mathrm{gen}[M_{123}]~,
\end{equation}
for some boundary regions $R_1,R_2,R_3$.
Let these regions be small and sufficiently separated such that the entanglement wedge of any joint region is the union of each constituent's entanglement wedge.
Furthermore, let $H_1, H_2,$ and $H_3$ be their respective homology hypersurfaces. 
In this case, the geometric part of MMI is saturated, and therefore MMI is satisfied for these boundary regions {\it if and only if} MMI is true for the three parties $H_1,H_2,H_3$: 
\begin{equation}
	S_{\mathrm{vN}}[H_1 H_2] + S_{\mathrm{vN}}[H_1 H_3] + S_{\mathrm{vN}}[H_2 H_3] \ge S_{\mathrm{vN}}[H_1] + S_{\mathrm{vN}}[H_2] + S_{\mathrm{vN}}[H_3] + S_{\mathrm{vN}}[H_1 H_2 H_3]~.	
\end{equation}
Bulk MMI is a necessary condition for boundary MMI to hold in general.

There is no reason to believe that MMI holds in general in the bulk.
For example, one can prepare four qubits in a four-party GHZ state
\begin{align}
\frac{1}{\sqrt{2}} \left( \ket{0000} + \ket{1111} \right).
\end{align}
MMI does not hold if each of the three parties are one of the parties in this state. 
We can therefore violate MMI by choosing very small boundary regions $R_i$, so that the bulk fields in each entanglement wedge are essentially uncorrelated in the vacuum state, and then placing one qubit from this four-party GHZ state in each of the three entanglement wedges.
It follows that holographic quantum states will not satisfy MMI in general, once quantum corrections are included.

For strong subaddivity, we could also prove a converse of this statement. Bulk SSA was a sufficient condition to imply boundary SSA. However, the same proof strategy that we used for SSA requires a much more complicated inequality than bulk MMI to prove boundary MMI.

By a simple generalization of the proof of entanglement wedge nesting, one can show that there exists a single Cauchy slice $C$ that contains the maximin surfaces $M_1$, $M_2$, $M_3$ and $M_{123}$ for the regions $R_1$, $R_2$, $R_3$ and $R_1 \cup R_2 \cup R_3$ and within which all those maximin surfaces have minimal generalized entropy.

The representatives $\tilde M_{12}$, $\tilde M_{13}$ and $\tilde M_{23}$ of the maximin surfaces for $R_1 \cup R_2$, $R_1 \cup R_3$ and $R_2 \cup R_3$  divide this Cauchy slice $C$ into \emph{eight} disjoint regions, labelled by whether or not they are contained in each of the three homology hypersurfaces $\tilde{H}_{ij}$. We label these regions $H_{\pm \pm \pm}$ where the three subscripts indicated whether the region is contained in ($+$) or not contained in ($-$) the homology hypersurfaces $\tilde{H}_{12}$, $\tilde{H}_{13}$ and $\tilde{H}_{23}$ respectively.

We can now construct new homology hypersurfaces
\begin{align}
    \begin{split}
        & H'_1 = \tilde{H}_{++-}, \\
        & H'_2 = \tilde{H}_{+-+}, \\
        & H'_3 = \tilde{H}_{-++}, \\
        H'_{123} = \tilde{H}_{+++} \cup \tilde{H}_{-++}& \cup \tilde{H}_{+-+} \cup \tilde{H}_{++-} \cup \tilde{H}_{--+} \cup \tilde{H}_{-+-} \cup \tilde{H}_{+--},
    \end{split}
\end{align}
which define surfaces $M_1'$, $M_2'$, $M_3'$ and $M_{123}'$ homologous to $R_1$, $R_2$, $R_3$ and $R_{1} \cup R_2 \cup R_3$ respectively.

By the minimality of the maximin surfaces within $C$, these new surfaces have larger generalized entropy than the maximin surfaces. By the quantum focussing conjecture, the representatives $\tilde{M}_{ij}$ have smaller generalized entropy than the maximin surfaces $M_{ij}$. Hence boundary MMI would follow if we had
\begin{align} \label{eq:needed}
	S_\mathrm{gen}(\tilde{M}_{12}) + S_\mathrm{gen}(\tilde{M}_{23}) + S_\mathrm{gen}(\tilde{M}_{13}) \geq S_\mathrm{gen}(M_{1}') + S_\mathrm{gen}(M_{2}') + S_\mathrm{gen}(M_{3}') + S_\mathrm{gen}(M_{123}').
\end{align}
The classical terms $S_\mathrm{grav}$ do indeed satisfy this inequality. Furthermore, if the representatives $\tilde{M}_{ij}$
have nontrivial transverse intersection, smoothing corners will decrease $S_\mathrm{grav}$ at leading order and MMI should hold so long as the bulk von Neumann entropies $S_\mathrm{vN}$ are subleading. However, by constructing a sufficiently complicated system of multiboundary wormholes, we can make the surfaces $\tilde{M}_{ij}$ not intersect, even though all eight regions $\tilde{H}_{\pm \pm \pm}$ are nonempty. In this case, the classical $S_\mathrm{grav}$ terms will exactly saturate the inequality \eqref{eq:needed}.

What about the von Neumann entropies? For \eqref{eq:needed} to hold we require
\begin{align} \label{eq:sevenregions}
\begin{split}
    S_\mathrm{vN}(\tilde{H}_{+++}& \cup \tilde{H}_{++-} \cup \tilde{H}_{+-+} \cup \tilde{H}_{+--}) + S_\mathrm{vN}(\tilde{H}_{+++} \cup \tilde{H}_{-++} \cup \tilde{H}_{++-} \cup \tilde{H}_{-+-})\\ +
    S_\mathrm{vN}&(\tilde{H}_{+++} \cup \tilde{H}_{-++} \cup \tilde{H}_{+-+} \cup \tilde{H}_{--+}) \geq S_\mathrm{vN}(\tilde{H}_{++-}) + S_\mathrm{vN}(\tilde{H}_{+-+}) + S_\mathrm{vN}(\tilde{H}_{-++}) \\&+ S_\mathrm{vN}(\tilde{H}_{+++} \cup \tilde{H}_{-++} \cup \tilde{H}_{+-+} \cup \tilde{H}_{++-} \cup \tilde{H}_{+--} \cup \tilde{H}_{-+-} \cup \tilde{H}_{--+}).
\end{split}
\end{align}
This is a significantly more complicated inequality than bulk MMI, involving seven regions rather than three.
However, in the special case where four of the seven regions are empty, it reduces to bulk MMI. It would be interesting to study whether \eqref{eq:sevenregions} is in fact implied by bulk MMI, or whether there exist quantum states that satisfy bulk MMI, for any choice of three regions, but do not satisfy \eqref{eq:sevenregions}. More generally, it would be interesting to know whether any inequality satisfied by the leading order holographic entropy cone extends to include quantum corrections so long as the same inequality is satisfied by the bulk fields themselves.

\section{Quantum Maximin for Nonholographic Quantum Subsystems} \label{sec:nonholographic}

It has recently been proposed that quantum extremal surfaces and entanglement wedges are well defined even for \emph{nonholographic} subsytems, such as a causal diamond in a quantum field theory, that are entangled with bulk degrees of freedom in a holographic theory or, more generally, for the combination of a holographic boundary region and additional entangled nonholographic subsystem. The nonholographic subsystem can even itself be in the bulk of a holographic theory, so long as it is in a bulk region where gravity can be ignored (for example near flat space asymptotic infinity). This idea was introduced in \cite{HayPen18} (for the combination of holographic and nonholographic systems) and \cite{Pen19} (for purely nonholographic systems) and was developed further (for purely nonholographic systems) in \cite{AMMZ,AMM}  where it was called the ``quantum extremal islands conjecture''.

Specifically, a quantum extremal surface $X_{R,Q}$ for the combination of a nonholographic subsystem $Q$ and a boundary region $R$ is defined to be a codimension-two surface, satisfying $X_{R,Q} \cup R = \partial H_{R,Q}$ for some acausal homology hypersurface $H_{R,Q}$, that is an extremum of
the generalized-entropy-like functional, which we name the hybrid entropy:
\begin{align} \label{gengenentropy}
    S_{\mathrm{hyb}}=S_{\mathrm{grav}}[X_{R,Q}] + S_{vN}[H_{R,Q} \cup Q],
\end{align}
where $S_{vN}[H_{R,Q} \cup Q]$ is the von Neumann entropy of the tensor product of $Q$ with the bulk fields in $H_{R,Q}$.\footnote{This formula generalises the formulas given in, for example, Eqn. 4.14 of \cite{HayPen18} and Eqn. 15 of ~\cite{AMMZ} to include higher derivative corrections.}  The entanglement wedge of the combination of $R$ and $Q$ is then defined to be the set of points which are completely determined by the data on $H_{R,Q} \cup Q$, which we shall call $D[H_{R,Q} \cup Q]$, in analogy with the domain of dependence, defined using the quantum extremal surface $X_{R,Q}$ that minimizes \eqref{gengenentropy}.\footnote{The original context for this rule was a derivation of the role of state dependence in entanglement wedge reconstruction. In this derivation, one has to consider bulk (code space) states that are entangled with an arbitrary (not necessarily holographic) reference system $Q$. State-independent reconstruction, on a boundary region $R$, is only possible if the bulk operator is not contained in $H_{\bar R, Q}$ for any such entangled state. (Here $\bar R$ is the complement of $R$.) This ends up being equivalent to the bulk operator being contained in $H_R$, for all \emph{mixed} states in the code space.}

If the boundary region $R$ is empty, then the homology constraint becomes $\partial H_{R,Q} = X_{R,Q}$. It follows $H_{R,Q}$ is an ``island'' in the bulk of the holographic theory, bounded entirely by $X_{R,Q}$. Somewhat remarkably, and unintuitively, if $Q$ contains the Hawking radiation of an old black hole, this island can become non-empty \cite{Pen19,AEMM}.

The conjecture in \cite{HayPen18,Pen19,AMMZ} is that the entropy of $Q \cup R$ is given by \eqref{gengenentropy}, evaluated on the minimal quantum extremal surface, and that bulk operators in the entanglement wedge can be reconstructed on $Q \cup R$. In \cite{HayPen18,Pen19}, this claim was justified by imagining throwing system $Q$ into the bulk of a very large holographic system, which cannot change the entropy, and then applying the standard quantum extremal surface prescription. 

In \cite{AMMZ}, an alternative perspective was proposed by considering the case where the bulk matter, and the `nonholographic' system $Q$ are both themselves holographic CFTs. In this case, it was hypothesized that the entropy of $Q$ can be found by using the classical HRT prescription in a `doubly holographic' description of the state, so long as the HRT surface is allowed to end on a surface in the original bulk geometry, which is interpreted as an end-of-the-world brane in this new description.

In most of the situations where quantum extremal surfaces have been considered for nonholographic systems, the system $Q$ has been the \emph{only} nonholographic system entangled with the holographic theory, and the overall state has been pure. In this case, the quantum  extremal surface has also been an ordinary quantum extremal surface for the complementary boundary region $\bar R$. However, when the state is mixed, or more than one nonholographic subsystem is entangled with the bulk fields, this is no longer true, and a genuinely new type of quantum extremal surface exists. See Section 3 of \cite{Pen19} for important cases where this is true. 

Our goal in this section is to define a quantum maximin prescription for the entropy of $Q \cup R$ and to show that it is equivalent to the minimal quantum extremal surface prescription discussed above, under the assumptions from Section~\ref{sec:3}. 

By doing so, we justify the argument from \cite{Pen19} that it is possible to  use maximin arguments to show that a non-empty island must exist for an old evaporating black hole, without having to actually find the quantum extremal surface.

We will then prove the same important consistency properties, such as nesting and SSA, for this prescription that we previously showed were satisfied for the standard quantum extremal surface prescription. 

The difference between these proofs and those of the previous sections arise because the quantity in~\eqref{gengenentropy} is not itself the generalized entropy of any \emph{bulk} region.
It is nearly the generalized entropy of the region $H_{R,Q}$, but with the additional nonholographic system $Q$ included in the entropy term. 
We will argue that nevertheless, the same assumptions can be used to construct a useful quantum maximin prescription for this entropy.

\subsection{Modifications of quantum maximin}
Let us begin with a formal definition of the quantum maximin prescription for the entropy of a nonholographic system $Q$, plus a boundary region $R$.
We define the quantum maximin surface $M_{R,Q}$ by the following maximinimization procedure
\begin{align} \label{genmaximin}
	\max_C \min_{M_{R,Q} \in C} \left[
	S_{\mathrm{grav}}[M_{R,Q}] + S_\mathrm{vN}[H_{R,Q} \cup Q] \right]~,
\end{align}
where the maximization is over time slices $C$ that are everywhere spacelike separated from the region $A$ and the minimization is over surfaces $M_{R,Q} \in C$ that satisfy $\partial H_{R,Q} = M_{R,Q} \cup R$ for some surface $\partial H_{R,Q}$. As before, $S_\mathrm{vN}[H_{R,Q} \cup Q]$ is the von Neumann entropy of $Q$ together with the bulk fields in $H_{R,Q}$. We will assume stability for $M_{R,Q}$ as well, where the definition is modified in the obvious way.

The proofs of existence etc. are all similar to their counterparts for ordinary quantum maximin surfaces.
Hence our strategy is to point out the differences between each of these proofs from the ones in sections~\ref{sec:existence} and \ref{sec:applications}.  

\subsubsection*{Existence}
One might worry that, since the hybrid entropy is not actually a generalized entropy, the arguments for the existence of maximin surface for the generalized entropy in Section \ref{sec:existence} may no longer apply.

However, maximinimizing the hybrid entropy entropy actually is the same as maximinimizing a generalized entropy. We simply use the trick from \cite{HayPen18,Pen19} of throwing the system $Q$ into an auxiliary holographic theory $S$ with sufficiently small gravitational coupling $G_N'$ (which should be different from the coupling $G_N$ in the original holographic system). Then \eqref{genmaximin} is simply the generalized entropy for the union of $H_{R,Q}$ and the entire bulk of $S$. Maximinimizing \eqref{genmaximin} is the same as doing ordinary maximin for the union of $R$ and the entire boundary of $S$, except that we aren't allowed to consider surfaces that have a nonempty component in the bulk of $S$. But if we take the limit $G_N' \to 0$ while keeping everything else fixed, the minimal generalized entropy surface in any Cauchy slice will never have a nonempty component in the bulk of $S$ and hence the maximin surface will be $M_{R,Q}$.

\subsubsection*{Equivalence with the Quantum Extremal Surface Prescription}
Crucial to the proof of equivalence is the use of the QFC to upperbound the generalized entropy of a representative surface by the original surface.
To do the same thing here, we now argue that the QFC implies the entropy Eqn.~(\ref{gengenentropy}) of a representative is upperbounded by the entropy Eqn.~(\ref{gengenentropy}) of the original surface.
As always, a representative $\tilde{X}$ of a codimension-two surface $X$ on a Cauchy slice $C$ is defined by releasing an orthogonal null congruence $N[X]$ from $X$ that intersects $C$, and defining $\tilde{X} = N[X] \cap C$. 
In general there are two representatives that fit this criterion -- we consider either one.

The application of the QFC is not completely immediate because the hybrid entropy is not a generalized entropy. However, it is still fairly straightforward. 
We just use our usual trick of throwing the system $Q$ into a gravitational theory and then take the limit where the gravitational coupling $G_N' \to 0$. We therefore conclude that a version of the QFC holds when $Q$ is a nongravitational theory and the generalized entropy is replaced with the hybrid entropy. 

This can also be argued more carefully by a separate analysis of each term in the quantum expansion. As explained in Section~\ref{sec:2}, ${\cal D}\Theta/{\cal D}\lambda$ has local terms that come from the gravitational entropy and the von Neumann entropy, and non-local terms that come only from the von Neumann entropy. Here ``local'' means proportional to a delta-function $\delta(y-y')$ or its derivatives, where $y$ and $y'$ are the locations of the two variations involved in ${\cal D}\Theta/{\cal D}\lambda$. The local, gravitational terms will be the same for variations of $M_{R,Q}$, regardless of whether the system includes $Q$ or not. The non-local terms in the von Neumann entropy are negative definite due to SSA, so they will not depend on whether $Q$ is non-gravitational. Only the more mysterious local terms from the von Neumann entropy could possibly care whether $Q$ is non-gravitational. The above argument suggests that they do not; one can construct a more careful argument using the fact that the diagonal terms usually amount to the statement of the so-called Quantum Null Energy Condition \cite{LeiLev18,BalCha19}. This condition depends only on data local to the point of variation. So, the inequality ${\cal D}\Theta/{\cal D}\lambda \le 0$ is just as likely to hold if a part of the system is non-gravitational, as long as that part is not being deformed. 

\subsubsection*{Entanglement Wedges Contain Causal Wedges}
Similar to the discussion above, we just need a version of the generalized second law (GSL) that applies to the hybrid entropy. We argue for such a law by exactly the same process of considering the hybrid entropy as the limit of a generalised entropy when a small gravitational coupling is removed.

\subsubsection*{Nesting}
We prove that the entanglement wedge, found using the quantum maximin prescription, does not increase in size when the boundary region $R$ is made smaller, and also when we only have access to a subsystem $Q'$ of the nonholographic quantum system $Q = Q'\otimes \bar{Q}'$; an example would be the restriction of a QFT state to a causal diamond $Q'$ that is entirely contained in the original causal diamond that defined $Q$. When $R$ is empty, this means that smaller nonholographic subsystems never have larger islands.
This is important in its own right, and also important for proving SSA in the following subsection.

\begin{thm}\label{thm:QEInesting}
	Let $R_1 \subseteq D[R_2]$ be a boundary region contained inside the domain of dependence $D[R_2]$ of the boundary region $R_2$ and let $Q_1$ be a subsystem of the nonholographic quantum system $Q_2 = Q_1 \otimes \bar{Q}_1$.
	Let $H_{R_1,Q_2}$ and $H_{R_2,Q_2}$ be the homology hypersurfaces associated to their respective quantum maximin surfaces $M_{R_1,Q_1}$ and $M_{R_2,Q_2}$.
	Then the domain of dependence of $H_{R_1,Q_1}$ is contained in that of $H_{R_2,Q_2}$, with $M_{R_1,Q_1}$ spacelike from $M_{R_2,Q_2}$.
	Furthermore, $M_{R_1,Q_1}$ and $M_{R_2,Q_2}$ are minimal on the same time slice.
\end{thm}

\begin{proof}
      This proof is nearly identical to the nesting proof in section \ref{sec:applications}. 
      The main difference is that the entropy inequalities now involve the additional systems $Q_1,Q_2$.
      We demonstrate that, nevertheless, nesting of the quantum maximin surface follows from strong subadditivity of the von Neumann entropy of the bulk fields and $Q_1,Q_2$ degrees of freedom. 
      
    Because of the similarity, here we only include the main points.
    We refer the reader to the proof of nesting in Section \ref{sec:applications} for additional details.
        We consider maximinimizing the quantity 
        $\alpha (S_\mathrm{grav}[M_1] + S_\mathrm{vN}[H_1 \cup Q_1]) + \beta (S_\mathrm{grav}[M_2] + S_\mathrm{vN}[H_2 \cup Q_2]))$, 
        for $M_{1}$ and $M_{2}$ acausal and homologous to $R_1, R_2$ respectively and $\alpha,\beta$ arbitrary positive real numbers.
        Here $H_1$ and $H_2$ are the homology hypersurfaces associated to $M_1$ and $M_2$ respectively.
        The surfaces $M_1$, $M_2$ found this way are both minimal hybrid entropy surfaces defined on the same time slice $C$. We will eventually show that they are in fact the quantum maximin surfaces $M_{R_1,Q_1}$ and $M_{R_2,Q_2}$.

We therefore want to show that $H_1 \subset H_2$. We will also need to show that $M_1 \cap M_2$ is a closed and open subset of $M_1$ and $M_2$ (i.e. they only intersect on entire connected components). To do so, there are exactly six types of bulk points we need to rule out, analogous to \eqref{pt:a}-\eqref{pt:f}.

Define $H_1' = \mathrm{Closure} (\mathrm{Int}(H_1) \cap \mathrm{Int}(H_2))$ and $H_2' = H_1 \cup H_2$ and define surfaces $M_1',M_2'$ by $\partial H_1' = R_1 \cup M_1'$ and $\partial H_2' = R_2 \cup M_2'$. See Figure~\ref{fig:HybridEWN}.

We now arrive at the primary difference between this proof and the one when $Q_1$ and $Q_2$ are trivial, as in Section~\ref{sec:applications}. 
We wish to write
\begin{align}
            S_{\mathrm{vN}}[H_1'\cup Q_1] + S_{\mathrm{vN}}[H_2' \cup Q_2] \le S_{\mathrm{vN}}[H_1 \cup Q_1] + S_{\mathrm{vN}}[H_2 \cup Q_2]~.
\end{align}
The analogous statement in Section~$\ref{sec:applications}$ followed from strong subadditivity of the bulk von Neumann entropy.
Here it follows from SSA of the von Neumann entropy of the bulk together with the systems $Q_1,Q_2$. 
That these systems together satisfy SSA is a very weak assumption -- so weak that it is hardly worth stating explicitly.

It is furthermore true that 
	\begin{align}\label{eqn:nestingareas_nonholographic}
            S_\mathrm{grav}[M_1'] + S_\mathrm{grav}[M_2'] \le S_\mathrm{grav}[M_1] + S_\mathrm{grav}[M_2]~.
    \end{align}
Note now that the surfaces $M_1',M_2'$ will in general have corners, which we must treat carefully because they have ill-defined extrinsic curvatures and therefore poorly defined higher derivative corrections to the geometric part of the generalized entropy. 
		To handle this, we define $M_1',M_2'$ with these corners ``smoothed out'' at a scale large relative to the Planck length and small compared to the bulk field theory scale.  
		As in section \ref{sec:applications}, we take this to reduce the geometric part of the generalized entropy without changing the renormalized von Neumann entropy part. 
		Hence, if there are any corners present then \eqref{eqn:nestingareas_nonholographic} is strict after smoothing.
		
 Therefore, in general
		\begin{align}	
			S_\mathrm{grav}[M_1'] &+ S_{\mathrm{vN}}[H_1'\cup Q_1] + S_\mathrm{grav}[M_2'] + S_{\mathrm{vN}}[H_2' \cup Q_2] \\
			&< S_\mathrm{grav}[M_1] + S_{\mathrm{vN}}[H_1 \cup Q_1] + S_\mathrm{grav}[M_2] + S_{\mathrm{vN}}[H_2 \cup Q_2]~,
		\end{align}
		where the inequality is strict if $(M_1',M_2') \neq (M_{1},M_{2})$.
		This contradicts the minimality of $M_{R',Q'}$ and $M_{R,Q}$, so it must be the case that $M_1' = M_{1}$ and $M_2' = M_{2}$.
		
Arguments directly analogous to those in section \ref{sec:applications} imply (1) $M_{R_1,Q_2}$ and $M_{R_2,Q_2}$ intersect only on entire connected components,
(2) points on $M_{R_1,Q_1}$ are not null-separated from points on $M_{R,Q}$, and therefore (3) $M_{R_1,Q_1}$ and $M_{R_2,Q_2}$ are both the minimal generalized entropy quantum extremal surfaces and hence maximin surfaces.
\end{proof}

\begin{figure}
\centering
\includegraphics{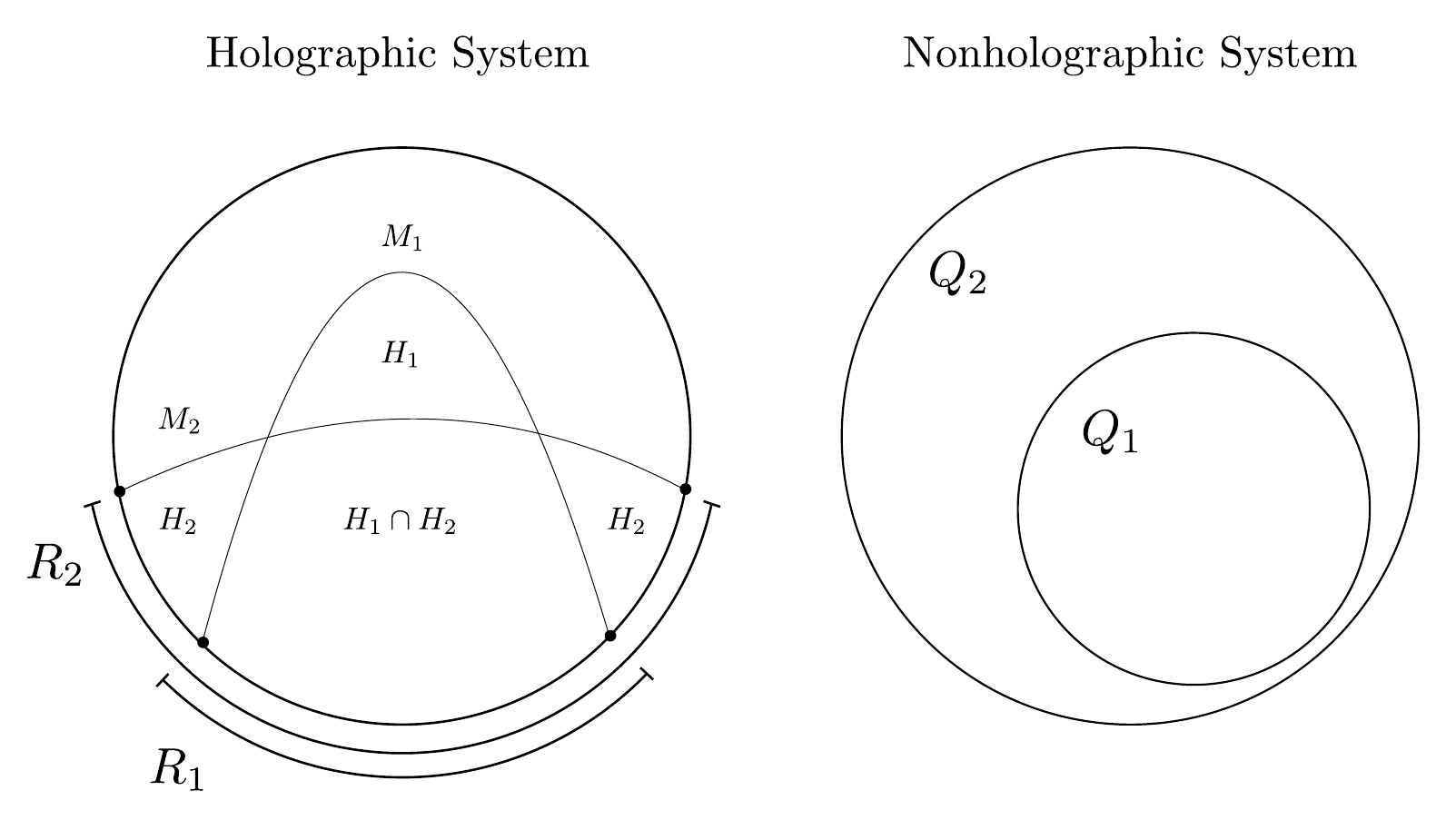}
\caption{The Hybrid entropy quantum maximin surfaces satisfy an appropriate extension of entanglement wedge nesting. It is shown by contradiction that we must have $H_{R_1,Q_1} \subset H_{R_2,Q_2}$ when $R_1 \subseteq D[R_2]$, and $Q_1 \subset Q_2$ thus proving entanglement wedge nesting for the hybrid entropy.}
\label{fig:HybridEWN}
\end{figure}

\subsubsection*{SSA}
We now prove that strong subadditivity is respected by quantum maximin surfaces of systems including nonholographic subsystems.

\begin{thm}
	Let $R_1$, $R_2$, $R_3$ be disjoint boundary regions and let $Q_1$, $Q_2$, $Q_3$ be nonholographic quantum subsystems.
	Let $M_{ij\dots}$ and $H_{ij\dots}$denote the quantum maximin surface and homology hypersurface respectively associated to the nonholographic subsystem $Q_i \otimes Q_j \dots$ together with the holographic boundary region $R_i \cup R_j \dots$.
	Then strong subadditivity holds:
	\begin{equation}
	\begin{split}
		\bigg(S_{\mathrm{grav}}[M_{12}&] + S_\mathrm{vN}[H_{12} \cup Q_1 Q_2]\bigg) + \bigg(S_{\mathrm{grav}}[M_{23}] + S_\mathrm{vN}[H_{23} \cup Q_2 Q_3]\bigg) 
		\ge \\ 
		\bigg(&S_{\mathrm{grav}}[M_{123}] + S_\mathrm{vN}[H_{123} \cup Q_1 Q_2 Q_3]\bigg) + \bigg(S_{\mathrm{grav}}[M_{2}] + S_\mathrm{vN}[H_{2} \cup Q_2]\bigg)~.
	\end{split}
	\end{equation}
\end{thm}

\begin{proof}
	Our strategy is almost identical to the SSA proof in Section~\ref{sec:applications}.
	We first find representatives of $M_{12}$ and $M_{23}$ on the same time slice on which $M_{123}$ and $M_{2}$ lie.
	Then the minimality of $M_{123}$ and $M_2$, combined with strong subadditivity of von Neumann entropy, will imply strong subadditivity of the quantum maximin hybrid entropy, as written above.

		By theorem \ref{thm:QEInesting}, $M_{123}$ and $M_{2}$ lie on the same time slice $C$.
		Both $M_{12}$ and $M_{23}$ have representatives $\widetilde{M}_{12}$ and $\widetilde{M}_{23}$ on $C$, which bound regions $\widetilde{H}_{12},\widetilde{H}_{23} \subset C$ with smaller hybrid entropy than $M_{12},M_{23}$ respectively. 
		Let $c_2 = \widetilde{H}_{12} \cap \widetilde{H}_{23}$ and $c_{123} = \widetilde{H}_{12} \cup \widetilde{H}_{23}$. 
			Strong subadditivity of the von Neumann entropy implies
			\begin{align}
				S_\mathrm{vN}[\widetilde{H}_{12} \cup Q_{1}Q_2] + S_\mathrm{vN}[\widetilde{H}_{23} \cup Q_{2}Q_3] \ge S_\mathrm{vN}[c_{123} \cup Q_1 Q_2 Q_3 ] + S_\mathrm{vN}[c_2 \cup Q_2]~.
			\end{align}
			Add to this the inequality 
			\begin{align}
			    S_\mathrm{grav}[\widetilde{H}_{12}]+ S_\mathrm{grav}[\widetilde{H}_{23}] \ge S_\mathrm{grav}[c_{123}] + S_\mathrm{grav}[c_2]~.
			\end{align}
		    Minimality of $H_{123}$ and $H_{2}$ on $C$ implies
		\begin{equation}
		\begin{split}
			\bigg( S_\mathrm{grav}[c_{123}]&+ S_\mathrm{vN}[c_{123} \cup Q_1 Q_2 Q_3] \bigg) + \bigg( S_\mathrm{grav}[c_{2}] + S_\mathrm{vN}[c_{2} \cup Q_2] \bigg) \\
			\ge& 
			\bigg( S_\mathrm{grav}[M_{123}] + S_\mathrm{vN}[H_{123} \cup Q_1 Q_2 Q_3] \bigg) + \bigg( S_\mathrm{grav}[M_{2}] + S_\mathrm{vN}[H_{2} \cup Q_2] \bigg)~.
		\end{split}
		\end{equation}
		Combining all of these inequalities concludes the proof.

\end{proof}

\section{Conclusion}

A complete description of the quantum behavior of black holes is at the crux of the black hole information paradox in particular and quantum gravity in general. The recent discoveries that new and fundamentally quantum physics manifest already under the inclusion of \textit{perturbative} quantum corrections to the geometry provides an exciting and powerful approach towards developing a better understanding of black hole information conservation, the firewall problem~\cite{AMPS, AMPSS}, spacetime emergence, and the resolution of the black hole singularity. In this paper, we have provided a new arsenal of tools for investigations of holographic black holes in the perturbatively quantum regime in the form of a quantum maximin reformulation of quantum extremal surfaces. The tools are powerful enough to also be of utility in deriving consistency conditions for evaporation-inspired modifications of holographic entanglement entropy formulae. 

Let us first concretely enumerate the results of this article before speculating on further applications. We have presented a maximin construction of the minimal quantum extremal surface (QES) anchored to a boundary subregion.
We have also presented a generalization, in which the extremized quantity includes the entropy of the bulk region (bounded by the surface) union a (possibly) non-holographic system. This encompasses the ``quantum extremal islands'' (QEI) proposal.
Using these constructions, we proved that QESs and QEIs obey nesting and strong subadditivity (SSA).
Both of these properties are important consistency checks if QES and QEI are to compute entropies. 
We also found that while bulk SSA manifestly implies boundary SSA, boundary MMI requires a seven party inequality in the bulk. This indicates potential for an interesting new investigation of the quantum holographic entropy cone, with potential to shed light about the connections between inequalities obeyed by the bulk and inequalities obeyed by the boundary theory.

We have made powerful but common assumptions throughout.
Foremost is the Quantum Focusing Conjecture (QFC), which we used to prove the equivalence of quantum maximin surfaces and QES (and quantum maximin islands to QEI). 
Separately we used it to prove nesting and SSA. The novelty of our application of the QFC is twofold: first, in its use in non-reflecting boundary conditions (see also~\cite{AMM}), and second, in our application of it to the hybrid entropy. 

The behavior of QESs (and now QEIs) surely holds part -- though not all -- of the key to understanding the bulk mechanism that implements unitarity. Ultimately, however, the QES describes the evolution of the entropy as dictated by the dynamics in the theory; a QES in a non-unitary theory would have to give a non-unitary evolution of the entropy (as illustrated in~\cite{AkersEngelhardtHarlow}): the EW prescription would be a poor prescription indeed if it consistently gave a unitary answer no matter the dynamics of the theory. However, more control and a better understanding of QESs gives us the power to ask: how must unitary dynamics of quantum gravity behave if the entropy evolution they dictate are to be described by quantum extremal surfaces?~\cite{AkersEngelhardtHarlow} gave such a toy model; the tools and understanding of QESs developed in this article may well pave the way for more comprehensive investigations and greater enlightenment about the quantum nature of black holes.



\section*{Acknowledgments}
It is a pleasure to thank D. Harlow, T. Hartman, P. Rath, A. Shahbazi-Moghaddam, and X. Zhou for discussions. CA is supported by the US Department of Energy grants DE-SC0018944 and DE-SC0019127, and also the Simons foundation as a member of the It from Qubit collaboration. NE is supported by the MIT department of physics and for part of this collaboration was supported by the Princeton Gravity Initiative and NSF grant No. PHY-1620059. GP is supported in part by AFOSR award FA9550-16-1- 0082 and DOE award {DE-SC0019380}. MU is supported in part by the National Science Foundation Graduate Research Fellowship Program under Grant No. DGE 1752814; by the Berkeley Center for Theoretical Physics; by the Department of Energy, Office of Science, Office of High Energy Physics under QuantISED Award DE-SC0019380 and under contract DE-AC02-05CH11231; and by the National Science Foundation under grant PHY1820912.

\appendix

\section{Transplanckian oscillations and the existence of minimal generalized entropy surfaces}\label{appendix}

In this section we elaborate on the question of existence of quantum maximin surfaces as discussed in Section \ref{sec:existence}. As argued for classical maximin in~\cite{Wal12}, the area functional on the space of surfaces, homologous to a boundary subregion $R$, on a single Cauchy slice is lower semi-continuous. That is, a small surface deformation can arbitrarily increase the area but cannot arbitrarily decrease it. This implies that a minimal area surface exists on any given Cauchy slice.

For quantum maximin surfaces, we would like to analogously show that a minimum generalized entropy surface, homologous to the appropriate boundary subregion, exists on all Cauchy slices. Thus the natural question is whether the generalized entropy is a lower semi-continuous functional on the space of relevant surfaces. If the surface has transplanckian fluctuations then higher curvature corrections contribute to $S_{\mathrm{grav}}$ at the same order as the area term, and so the functional need not be lower semi-continuous. In fact, we will argue that $S_{\mathrm{grav}}$ will not converge for such surfaces, and so it is unclear if generalized entropy is well defined for surfaces with transplanckian fluctuations.

We will work with a specific Lagrangian and construct a surface deformation that naively reduces $S_{\mathrm{grav}}$. Any initial surface can then be deformed everywhere by the constructed deformation, which would decrease the entropy by a macroscopic amount. From this we can conclude that minimization over all surfaces will naively not converge to the extremal surface, instead favoring a surface with sharp oscillations. The resolution to this is that we are unfairly truncating the Lagrangian and gravitational entropy functional. We should include all possible higher derivative terms in the action with appropriate EFT couplings. When this is done we argue that the entropy functional will not converge, and so when minimizing over all surfaces we should not take surfaces with transplanckian fluctuations seriously.

We begin by introducing some notation. A codimension-two surface has $d-2$ tangent coordinates $y^i$ with embedding coordinates $X^\mu(y^i)$. The surface has induced metric $h_{ij}$ with associated Christoffel symbols $\gamma_{ij}^{k}$. Surfaces indices will be denoted by $i,j,k$ while $\mu, \nu, \sigma$ will be bulk spacetime indices. The bulk metric is $g_{\mu \nu}$ with Christoffel symbols $\Gamma^{\mu}_{\nu \sigma}$. The extrinsic curvature and normal projector to the surface are~\cite{AkeCha17}

\begin{equation}
    K^{\mu}_{ij} = \partial_{i} \partial_{j} X^{\mu} + \gamma_{ij}^{k} \partial_{k} X^{\mu} - \Gamma^{\mu}_{\nu \sigma} \partial_{i} X^{\nu} \partial_{j} X^{\sigma},
\end{equation}

\begin{equation}
    K^\mu = h^{ij}  K^{\mu}_{ij},
\end{equation}

\begin{equation}
    N_{\mu \nu} = g_{\mu \nu} - (\partial_{i} X_{\mu}) (\partial_{j} X_{\nu}) h^{ij}.
\end{equation}

\noindent We'll work with the following simple Lagrangian with $\lambda>0$ an order one dimensionless coupling and $l$ the EFT length scale, and we take our gravitational entropy functional to be the one computed by~\cite{Don13}:

\begin{equation}
    \mathcal{L} = \frac{1}{16 \pi G_{N}} (R + \lambda l^{2} R_{\mu \nu}^{2}),
\end{equation}

\begin{equation}
    S_{\mathrm{grav}} = \frac{1}{4 G_{N}} \int \sqrt{|h|} (1+ \lambda l^{2} N_{\mu \nu}(R^{\mu \nu}-\frac{1}{2} K^{\mu} K^{\nu})).
\end{equation}

For simplicity, we work in $2+1$d Minkowski space and take our initial surface to be a line. We take our surface deformation to be a sharp bump function. The embedding coordinates are $X^{\mu} (y) = (0,y,f(y))$ with $f(y)=L_{1} \exp (\frac{-\alpha}{1-y^{2}/L_{2}^{2}})$. $L_{1}$ and $L_{2}$ are the two length scales of the deformation and $\alpha$ is a dimensionless parameter. We find the extrinsic curvature term to be

\begin{equation}
    N_{\mu \nu}K^\mu K^\nu = \frac{\ddot{f}^2}{(1+\dot{f}^2)^3}.
\end{equation}

\noindent Since the derivative of the bump function is small we can expand the induced metric in the entropy functional

\begin{equation}
S_{\mathrm{grav}} \approx 
    \frac{1}{4 G_{N}} \int \sqrt{|h|} (1+\frac{1}{2}\dot{f}^2) (1-\frac{1}{2} \lambda l^2 K^2) \approx \frac{1}{4 G_{N}} \int \sqrt{|h|} (1+\frac{1}{2}\dot{f}^2 -\frac{1}{2}\lambda l^2 K^2),
\end{equation}

\noindent where we have dropped higher order terms. From the above it is clear that $S_{\text{Dong}}$ decreases when $\dot{f}^2 - \lambda l^{2} K^{2} < 0$. We set $\lambda=1$ to simplify the analysis and we work in units where the EFT length scale is $l=1$. Now since the bump function always has an inflection point where the extrinsic curvature vanishes we cannot always satisfy this inequality. However, there exist profiles that on average satisfy this inequality. As can be checked numerically for the exact functional, $L_{1} = L_{2} =5$ and $\alpha = 10$ gives a bump  function deformation of the line that decreases the entropy, and we expect such deformations to generically exist for all initial surfaces. Note that in this example the extrinsic curvature of the bump is small.

The resolution to the above is that we have ignored all higher order terms. While higher order corrections will be suppressed by additional factors of the EFT length scale $l$, they can also involve higher derivatives of the extrinsic curvature, which corresponds to higher derivatives of $f$. As an example of such a term, adding $\nabla_{\mu}R_{\nu \sigma} \nabla^{\mu}R^{\nu \sigma}$ to the Lagrangian will modify the entropy functional to include~\cite{GuoMia14} $(\nabla K)^2$, where we have omitted indices to simplify notation. 

For deformations which oscillate on scales that are smaller than the EFT scale, we expect that we will find a tower of corrections, each larger than the one before. Thus, we cannot say that the deformation 'decreases' $S_{\mathrm{gen}}$ since the entropy functional does not even converge.  

We expect surface deformations without such problems to never be able to decrease the generalized entropy by an unbounded amount, and so we expect a well behaved minimal generalized entropy surface to exist. 

\bibliographystyle{jhep}
\bibliography{all}

\providecommand{\href}[2]{#2}\begingroup\raggedright\begin{thebibliography}{10}

\bibitem{CzeKar12}
B.~Czech, J.~L. Karczmarek, F.~Nogueira, and M.~Van~Raamsdonk, {\it {The
  Gravity Dual of a Density Matrix}},  {\em Class.Quant.Grav.} {\bf 29} (2012)
  155009, [\href{http://arxiv.org/abs/1204.1330}{{\tt arXiv:1204.1330}}].

\bibitem{Wal12}
A.~C. Wall, {\it {Maximin Surfaces, and the Strong Subadditivity of the
  Covariant Holographic Entanglement Entropy}},  {\em Class.Quant.Grav.} {\bf
  31} (2014), no.~22 225007, [\href{http://arxiv.org/abs/1211.3494}{{\tt
  arXiv:1211.3494}}].

\bibitem{DonHar16}
X.~Dong, D.~Harlow, and A.~C. Wall, {\it {Reconstruction of Bulk Operators
  within the Entanglement Wedge in Gauge-Gravity Duality}},  {\em Phys. Rev.
  Lett.} {\bf 117} (2016), no.~2 021601,
  [\href{http://arxiv.org/abs/1601.05416}{{\tt arXiv:1601.05416}}].

\bibitem{AlmDon14}
A.~Almheiri, X.~Dong, and D.~Harlow, {\it {Bulk Locality and Quantum Error
  Correction in AdS/CFT}},  {\em JHEP} {\bf 04} (2015) 163,
  [\href{http://arxiv.org/abs/1411.7041}{{\tt arXiv:1411.7041}}].

\bibitem{Har16}
D.~Harlow, {\it {The Ryu-Takayanagi Formula from Quantum Error Correction}},
  {\em Commun. Math. Phys.} {\bf 354} (2017), no.~3 865--912,
  [\href{http://arxiv.org/abs/1607.03901}{{\tt arXiv:1607.03901}}].

\bibitem{HaPPY}
F.~Pastawski, B.~Yoshida, D.~Harlow, and J.~Preskill, {\it {Holographic quantum
  error-correcting codes: Toy models for the bulk/boundary correspondence}},
  {\em JHEP} {\bf 06} (2015) 149, [\href{http://arxiv.org/abs/1503.06237}{{\tt
  arXiv:1503.06237}}].

\bibitem{HayPen18}
P.~Hayden and G.~Penington, {\it {Learning the Alpha-bits of Black Holes}},
  \href{http://arxiv.org/abs/1807.06041}{{\tt arXiv:1807.06041}}.

\bibitem{AkeLei19}
C.~Akers, S.~Leichenauer, and A.~Levine, {\it Large breakdowns of entanglement
  wedge reconstruction},  {\em arXiv preprint arXiv:1908.03975} (2019).

\bibitem{Pen19}
G.~Penington, {\it {Entanglement Wedge Reconstruction and the Information
  Paradox}},  \href{http://arxiv.org/abs/1905.08255}{{\tt arXiv:1905.08255}}.

\bibitem{AEMM}
A.~Almheiri, N.~Engelhardt, D.~Marolf, and H.~Maxfield, {\it {The entropy of
  bulk quantum fields and the entanglement wedge of an evaporating black
  hole}},  \href{http://arxiv.org/abs/1905.08762}{{\tt arXiv:1905.08762}}.

\bibitem{EngWal14}
N.~Engelhardt and A.~C. Wall, {\it {Quantum Extremal Surfaces: Holographic
  Entanglement Entropy beyond the Classical Regime}},  {\em JHEP} {\bf 01}
  (2015) 073, [\href{http://arxiv.org/abs/1408.3203}{{\tt arXiv:1408.3203}}].

\bibitem{Bek72}
J.~D. Bekenstein, {\it Black holes and the second law},  {\em Nuovo Cim. Lett.}
  {\bf 4} (1972) 737--740.

\bibitem{IyeWal94}
V.~Iyer and R.~M. Wald, {\it Some properties of {N}oether charge and a proposal
  for dynamical black hole entropy},  {\em Phys. Rev. D} {\bf 50} (1994)
  846--864,
  [\href{http://arxiv.org/abs/http://arXiv.org/abs/gr-qc/9403028}{{\tt
  http://arXiv.org/abs/gr-qc/9403028}}].

\bibitem{Lei17}
S.~Leichenauer, {\it {The Quantum Focusing Conjecture Has Not Been Violated}},
  \href{http://arxiv.org/abs/1705.05469}{{\tt arXiv:1705.05469}}.

\bibitem{AkeCha17}
C.~Akers, V.~Chandrasekaran, S.~Leichenauer, A.~Levine, and
  A.~Shahbazi~Moghaddam, {\it {The Quantum Null Energy Condition, Entanglement
  Wedge Nesting, and Quantum Focusing}},
  \href{http://arxiv.org/abs/1706.04183}{{\tt arXiv:1706.04183}}.

\bibitem{Wal10Proofs}
A.~C. Wall, {\it {Ten Proofs of the Generalized Second Law}},  {\em JHEP} {\bf
  06} (2009) 021, [\href{http://arxiv.org/abs/0901.3865}{{\tt
  arXiv:0901.3865}}].

\bibitem{RyuTak06}
S.~Ryu and T.~Takayanagi, {\it {Holographic derivation of entanglement entropy
  from AdS/CFT}},  {\em Phys.Rev.Lett.} {\bf 96} (2006) 181602,
  [\href{http://arxiv.org/abs/hep-th/0603001}{{\tt hep-th/0603001}}].

\bibitem{HubRan07}
V.~E. Hubeny, M.~Rangamani, and T.~Takayanagi, {\it {A Covariant holographic
  entanglement entropy proposal}},  {\em JHEP} {\bf 0707} (2007) 062,
  [\href{http://arxiv.org/abs/0705.0016}{{\tt arXiv:0705.0016}}].

\bibitem{MarWal19}
D.~Marolf, A.~C. Wall, and Z.~Wang, {\it {Restricted Maximin surfaces and HRT
  in generic black hole spacetimes}},
  \href{http://arxiv.org/abs/1901.03879}{{\tt arXiv:1901.03879}}.

\bibitem{HeaTak}
M.~Headrick and T.~Takayanagi, {\it {A Holographic proof of the strong
  subadditivity of entanglement entropy}},  {\em Phys.Rev.} {\bf D76} (2007)
  106013, [\href{http://arxiv.org/abs/0704.3719}{{\tt arXiv:0704.3719}}].

\bibitem{EngHarTA}
N.~Engelhardt and D.~Harlow, to appear.

\bibitem{BouCha19}
R.~Bousso, V.~Chandrasekaran, and A.~Shahbazi-Moghaddam, {\it Ignorance is
  cheap: From black hole entropy to energy-minimizing states in qft},  {\em
  arXiv preprint arXiv:1906.05299} (2019).

\bibitem{BouFis15}
R.~Bousso, Z.~Fisher, S.~Leichenauer, and A.~C. Wall, {\it {Quantum focusing
  conjecture}},  {\em Phys. Rev.} {\bf D93} (2016), no.~6 064044,
  [\href{http://arxiv.org/abs/1506.02669}{{\tt arXiv:1506.02669}}].

\bibitem{HayHeaMal13}
P.~Hayden, M.~Headrick, and A.~Maloney, {\it Holographic mutual information is
  monogamous},  {\em Physical Review D} {\bf 87} (2013), no.~4 046003.

\bibitem{AMMZ}
A.~Almheiri, R.~Mahajan, J.~Maldacena, and Y.~Zhao, {\it {The Page curve of
  Hawking radiation from semiclassical geometry}},
  \href{http://arxiv.org/abs/1908.10996}{{\tt arXiv:1908.10996}}.

\bibitem{PenShe19}
G.~Penington, S.~H. Shenker, D.~Stanford, and Z.~Yang, {\it {Replica wormholes
  and the black hole interior}},  \href{http://arxiv.org/abs/1911.11977}{{\tt
  arXiv:1911.11977}}.

\bibitem{AlmHar19}
A.~Almheiri, T.~Hartman, J.~Maldacena, E.~Shaghoulian, and A.~Tajdini, {\it
  {Replica Wormholes and the Entropy of Hawking Radiation}},
  \href{http://arxiv.org/abs/1911.12333}{{\tt arXiv:1911.12333}}.

\bibitem{AlmMahSan19}
A.~Almheiri, R.~Mahajan, and J.~E. Santos, {\it {Entanglement islands in higher
  dimensions}},  \href{http://arxiv.org/abs/1911.09666}{{\tt
  arXiv:1911.09666}}.

\bibitem{Wald}
R.~M. Wald, {\em General Relativity}.
\newblock The University of Chicago Press, Chicago, 1984.

\bibitem{EngFis19}
N.~Engelhardt and S.~Fischetti, {\it Surface theory: the classical, the
  quantum, and the holographic},  {\em arXiv preprint arXiv:1904.08423} (2019).

\bibitem{Haw75}
S.~W. Hawking, {\it Particle creation by black holes},  {\em Commun. Math.
  Phys.} {\bf 43} (1975) 199.

\bibitem{EpsGla65}
H.~Epstein, V.~Glaser, and A.~Jaffe, {\it Nonpositivity of the energy density
  in quantized field theories},  {\em Il Nuovo Cimento Series 10} {\bf 36}
  (1965) 1016--1022.

\bibitem{Cas48}
H.~B.~G. Casimir, {\it {On the Attraction Between Two Perfectly Conducting
  Plates}},  {\em Indag. Math.} {\bf 10} (1948) 261--263. [Kon. Ned. Akad.
  Wetensch. Proc.100N3-4,61(1997)].

\bibitem{Dav76}
P.~C.~W. Davies and S.~A. Fulling, {\it {Radiation from a moving mirror in
  two-dimensional space-time conformal anomaly}},  {\em Proc. Roy. Soc. Lond.}
  {\bf A348} (1976) 393--414.

\bibitem{DavUnr76}
P.~C.~W. Davies, S.~A. Fulling, and W.~G. Unruh, {\it {Energy Momentum Tensor
  Near an Evaporating Black Hole}},  {\em Phys. Rev.} {\bf D13} (1976)
  2720--2723.

\bibitem{Dav77}
P.~C.~W. Davies and S.~A. Fulling, {\it {Radiation from Moving Mirrors and from
  Black Holes}},  {\em Proc. Roy. Soc. Lond.} {\bf A356} (1977) 237--257.

\bibitem{Bek73}
J.~D. Bekenstein, {\it Black holes and entropy},  {\em Phys. Rev. D} {\bf 7}
  (1973) 2333.

\bibitem{Wal10QST}
A.~C. Wall, {\it {The Generalized Second Law implies a Quantum Singularity
  Theorem}},  {\em Class.Quant.Grav.} {\bf 30} (2013) 165003,
  [\href{http://arxiv.org/abs/1010.5513}{{\tt arXiv:1010.5513}}].

\bibitem{BouEng15c}
R.~Bousso and N.~Engelhardt, {\it {Generalized Second Law for Cosmology}},
  {\em Phys. Rev.} {\bf D93} (2016), no.~2 024025,
  [\href{http://arxiv.org/abs/1510.02099}{{\tt arXiv:1510.02099}}].

\bibitem{Don13}
X.~Dong, {\it {Holographic Entanglement Entropy for General Higher Derivative
  Gravity}},  {\em JHEP} {\bf 01} (2014) 044,
  [\href{http://arxiv.org/abs/1310.5713}{{\tt arXiv:1310.5713}}].

\bibitem{GuoMia14}
R.-X. Miao and W.-z. Guo, {\it Holographic entanglement entropy for the most
  general higher derivative gravity},  {\em Journal of High Energy Physics}
  {\bf 2015} (2015), no.~8 31.

\bibitem{Cam13}
J.~Camps, {\it {Generalized entropy and higher derivative Gravity}},  {\em
  JHEP} {\bf 03} (2014) 070, [\href{http://arxiv.org/abs/1310.6659}{{\tt
  arXiv:1310.6659}}].

\bibitem{DonLew17}
X.~Dong and A.~Lewkowycz, {\it {Entropy, extremality, euclidean variations, and
  the equations of motion}},  \href{http://arxiv.org/abs/1705.08453}{{\tt
  arXiv:1705.08453}}.

\bibitem{DonMar19}
X.~Dong and D.~Marolf, {\it {One-loop universality of holographic codes}},
  \href{http://arxiv.org/abs/1910.06329}{{\tt arXiv:1910.06329}}.

\bibitem{FurSol94}
D.~V. Fursaev and S.~N. Solodukhin, {\it {On one loop renormalization of black
  hole entropy}},  {\em Phys. Lett.} {\bf B365} (1996) 51--55,
  [\href{http://arxiv.org/abs/hep-th/9412020}{{\tt hep-th/9412020}}].

\bibitem{Gal99}
G.~J. Galloway, {\it {Maximum principles for null hypersurfaces and null
  splitting theorems}},  {\em Annales Henri Poincare} {\bf 1} (2000) 543--567,
  [\href{http://arxiv.org/abs/math/9909158}{{\tt math/9909158}}].

\bibitem{BouFis15b}
R.~Bousso, Z.~Fisher, J.~Koeller, S.~Leichenauer, and A.~C. Wall, {\it {Proof
  of the Quantum Null Energy Condition}},  {\em Phys. Rev.} {\bf D93} (2016),
  no.~2 024017, [\href{http://arxiv.org/abs/1509.02542}{{\tt
  arXiv:1509.02542}}].

\bibitem{KoeLei15}
J.~Koeller and S.~Leichenauer, {\it {Holographic Proof of the Quantum Null
  Energy Condition}},  {\em Phys. Rev.} {\bf D94} (2016), no.~2 024026,
  [\href{http://arxiv.org/abs/1512.06109}{{\tt arXiv:1512.06109}}].

\bibitem{FuKoe17b}
Z.~Fu, J.~Koeller, and D.~Marolf, {\it {The Quantum Null Energy Condition in
  Curved Space}},  {\em Class. Quant. Grav.} {\bf 34} (2017), no.~22 225012,
  [\href{http://arxiv.org/abs/1706.01572}{{\tt arXiv:1706.01572}}]. [Erratum:
  Class. Quant. Grav.35,no.4,049501(2018)].

\bibitem{BalFau17}
S.~Balakrishnan, T.~Faulkner, Z.~U. Khandker, and H.~Wang, {\it {A General
  Proof of the Quantum Null Energy Condition}},
  \href{http://arxiv.org/abs/1706.09432}{{\tt arXiv:1706.09432}}.

\bibitem{LeiLev18}
S.~Leichenauer, A.~Levine, and A.~Shahbazi-Moghaddam, {\it {Energy density from
  second shape variations of the von Neumann entropy}},  {\em Phys. Rev.} {\bf
  D98} (2018), no.~8 086013, [\href{http://arxiv.org/abs/1802.02584}{{\tt
  arXiv:1802.02584}}].

\bibitem{BalCha19}
S.~Balakrishnan, V.~Chandrasekaran, T.~Faulkner, A.~Levine, and
  A.~Shahbazi-Moghaddam, {\it {Entropy Variations and Light Ray Operators from
  Replica Defects}},  \href{http://arxiv.org/abs/1906.08274}{{\tt
  arXiv:1906.08274}}.

\bibitem{FuKoe17}
Z.~Fu, J.~Koeller, and D.~Marolf, {\it {Violating the quantum focusing
  conjecture and quantum covariant entropy bound in $d\ge 5$ dimensions}},
  {\em Class. Quant. Grav.} {\bf 34} (2017), no.~17 175006,
  [\href{http://arxiv.org/abs/1705.03161}{{\tt arXiv:1705.03161}}].

\bibitem{AMM}
A.~Almheiri, R.~Mahajan, and J.~Maldacena, {\it Islands outside the horizon},
  \href{http://arxiv.org/abs/1910.11077}{{\tt arXiv:1910.11077}}.

\bibitem{BouCas14a}
R.~Bousso, H.~Casini, Z.~Fisher, and J.~Maldacena, {\it {Proof of a Quantum
  Bousso Bound}},  {\em Phys.Rev.} {\bf D90} (2014) 044002,
  [\href{http://arxiv.org/abs/1404.5635}{{\tt arXiv:1404.5635}}].

\bibitem{BouCas14b}
R.~Bousso, H.~Casini, Z.~Fisher, and J.~Maldacena, {\it {Entropy on a null
  surface for interacting quantum field theories and the Bousso bound}},
  \href{http://arxiv.org/abs/1406.4545}{{\tt arXiv:1406.4545}}.

\bibitem{bombieri2011geometric}
E.~Bombieri, {\em Geometric Measure Theory and Minimal Surfaces: Lectures given
  at a Summer School of the Centro Internazionale Matematico Estivo (C.I.M.E.)
  held in Varenna (Como), Italy, August 24 - September 2, 1972}.
\newblock C.I.M.E. Summer Schools. Springer Berlin Heidelberg, 2011.

\bibitem{Sor19}
J.~Sorce, {\it Holographic entanglement entropy is cutoff-covariant},  {\em
  Journal of High Energy Physics} {\bf 2019} (2019), no.~10 15.

\bibitem{FisMar14}
S.~Fischetti, D.~Marolf, and A.~C. Wall, {\it {A paucity of bulk entangling
  surfaces: AdS wormholes with de Sitter interiors}},  {\em Class.Quant.Grav.}
  {\bf 32} (2015), no.~6 065011, [\href{http://arxiv.org/abs/1409.6754}{{\tt
  arXiv:1409.6754}}].

\bibitem{BouEng15b}
R.~Bousso and N.~Engelhardt, {\it {Proof of a New Area Law in General
  Relativity}},  {\em Phys. Rev.} {\bf D92} (2015), no.~4 044031,
  [\href{http://arxiv.org/abs/1504.07660}{{\tt arXiv:1504.07660}}].

\bibitem{BaoNez15}
N.~Bao, S.~Nezami, H.~Ooguri, B.~Stoica, J.~Sully, and M.~Walter, {\it {The
  Holographic Entropy Cone}},  {\em JHEP} {\bf 09} (2015) 130,
  [\href{http://arxiv.org/abs/1505.07839}{{\tt arXiv:1505.07839}}].

\bibitem{AMPS}
A.~Almheiri, D.~Marolf, J.~Polchinski, and J.~Sully, {\it {Black Holes:
  Complementarity or Firewalls?}},  \href{http://arxiv.org/abs/1207.3123}{{\tt
  arXiv:1207.3123}}.

\bibitem{AMPSS}
A.~Almheiri, D.~Marolf, J.~Polchinski, D.~Stanford, and J.~Sully, {\it {An
  Apologia for Firewalls}},  \href{http://arxiv.org/abs/1304.6483}{{\tt
  arXiv:1304.6483}}.

\bibitem{AkersEngelhardtHarlow}
C.~Akers, N.~Engelhardt, and D.~Harlow, {\it {Simple holographic models of
  black hole evaporation}},  \href{http://arxiv.org/abs/1910.00972}{{\tt
  arXiv:1910.00972}}.

\end{thebibliography}\endgroup

\end{document}